\documentclass[11pt]{article}

\usepackage[
  bookmarks=true,
  bookmarksnumbered=true,
  bookmarksopen=true,
  pdfborder={0 0 0},
  breaklinks=true,
  colorlinks=true,
  linkcolor=black,
  citecolor=black,
  filecolor=black,
  urlcolor=black,
]{hyperref}

\usepackage{amsthm}
\usepackage{subfigure}
\usepackage{amsmath}
\usepackage{amsfonts}
\usepackage{amssymb}
\usepackage{graphicx}
\usepackage[capitalise]{cleveref}

\usepackage{nicefrac}

\usepackage[margin=1in]{geometry}

\newtheorem{theorem}{Theorem}
\newtheorem{proposition}{Proposition}
\newtheorem{lemma}{Lemma}
\newtheorem{conjecture}{Conjecture}

\theoremstyle{definition}
\newtheorem{definition}{Definition}
\newtheorem{remark}{Remark}

\newcommand{\sfa}{\mathsf{a}}
\newcommand{\mv}{\mathtt{v}}
\newcommand{\existsv}{\exists{\mv}}
\newcommand{\existsZ}{\exists{0}}
\newcommand{\knownf}[1]{\ensuremath{F\node{#1}}}
\newcommand{\minval}[1]{\ensuremath{\mathit{Min}\node{#1}}}
\newcommand{\knownvals}[1]{\ensuremath{\mathit{Vals}\node{#1}}}
\newcommand{\knownlows}[1]{\ensuremath{\mathit{Lows}\node{#1}}}
\newcommand{\eqdef}{\triangleq}

\newcommand{\OptMink}{\mbox{$\mbox{\sc Opt}_{\min[k]}$}}

\newcommand{\set}[1]{\{#1\}}

\newcommand{\Crash}{\mathsf{Crash}}
\def\emptyset{\mbox{\O}}

\newcommand{\defemph}[1]{\textbf{\textit{#1}}}
\newcommand{\sat}{\models}

\newcommand{\decide}{\mathsf{decide}}

\newcommand{\Proc}{\mathsf{Procs}}
\newcommand{\tee}{\,\defemph{t}}

\newcommand{\OptZ}{\mbox{$\mbox{{\sc Opt}}_0$}}
\newcommand{\UOptZ}{\mbox{$\mbox{{\sc u-Opt}}_0$}}
\newcommand{\UOptMink}{\mbox{$\mbox{\sc u-P}_{\min[k]}$}}

\newcommand{\Fmodel}{{\cal F}}
\newcommand{\CG}{{\cal G}}
\newcommand{\Vals}{{\tt V}}
\newcommand{\Vecs}{\vec{\Vals}}
\newcommand{\dom}{\,{\preceq}\,}

\newcommand{\pdo}{unbeatable}
\newcommand{\FP}{\mathsf{F}}

\newcommand{\node}[1]{\langle#1\rangle}
\newcommand{\Ga}{\CG_\alpha}

\newcommand{\cw}{\exists\mathsf{correct}(\mathtt{w})}
\newcommand{\cv}{\exists\mathsf{correct}(\mathtt{v})}

\DeclareMathOperator{\Bary}{\sf Bary}
\DeclareMathOperator{\Star}{\sf St}
\DeclareMathOperator{\Car}{\sf Car}
\DeclareMathOperator{\Div}{\sf Div}
\DeclareMathOperator{\bdry}{\sf Bd}

\newcommand{\ang}[1]{\langle{#1}\rangle}

\newcommand{\cK}{\ensuremath{\mathcal{K}}}
\newcommand{\cL}{\ensuremath{\mathcal{L}}}
\newcommand{\cP}{\ensuremath{\mathcal{P}}}

\newcommand{\kAgreement}{{\bf \defemph{k}-Agreement}}
\newcommand{\Agreement}{{\bf Agreement}}
\newcommand{\UnikAg}{{\bf Uniform \defemph{k}-Agreement}}

\newcommand{\Decision}{{\bf Decision}}
\newcommand{\Validity}{{\bf Validity}}

\newcommand{\valv}{\veee}
\newcommand{\veee}{\mathtt{v}}

\newcommand{\decideiZ}{\mathsf{decide(0)}}

\newcommand{\decideiO}{\mathsf{decide(1)}}
\newcommand{\fip}{{\it fip}}

\newcommand{\HC}[1]{\mathsf{HC}\node{#1}}

\hypersetup{
  pdfauthor      = {Armando Casta{\~n}eda <armando@im.unam.mx>, Yannai A. Gonczarowski <yannai@gonch.name>, and Yoram Moses <moses@ee.technion.ac.il>},
  pdftitle       = {Unbeatable Set Consensus via Topological and Combinatorial Reasoning},
}

\begin{document}

\title{Unbeatable Set Consensus via \\ Topological and Combinatorial Reasoning\thanks{
Part of the results of this paper
were announced
in~\cite{AYY-PODC-BA}.
}}

\author{Armando Casta\~{n}eda\thanks{Universidad Nacional Aut\'onoma de M\'exico (UNAM), \mbox{\emph{E-mail}: \href{mailto:armando@im.unam.mx}{armando@im.unam.mx}}.}
\and Yannai A.~Gonczarowski\thanks{The Hebrew University of Jerusalem and Microsoft Research, \mbox{\emph{E-mail}: \href{mailto:yannai@gonch.name}{yannai@gonch.name}}.}
\and Yoram Moses\thanks{Technion --- Israel Institute of Technology, \mbox{\emph{E-mail}: \href{mailto:moses@ee.technion.ac.il}{moses@ee.technion.ac.il}}.}}

\date{May 23, 2016}

\maketitle
\begin{abstract}
The \defemph{set consensus} problem has played an important role in the study of distributed systems for over two decades. Indeed, the search for lower bounds and impossibility results for this problem  spawned the topological approach to distributed computing, which has given rise to new techniques in the design and analysis of protocols. The design of efficient solutions to set consensus has also proven to be challenging.
In the synchronous crash failure model, the literature contains a sequence of solutions to set consensus, each improving upon the previous ones.

This paper presents an \defemph{unbeatable} protocol for nonuniform $k$-set consensus in the synchronous crash failure model.
This is an efficient protocol whose decision times \emph{cannot} be improved upon.
Moreover, the description of our protocol is extremely succinct.
Proving unbeatability of this protocol is a nontrivial challenge.
We provide two proofs for its unbeatability: one is a subtle constructive combinatorial proof, and the other is a topological proof of a new style.
These two proofs provide new insight into the connection between  topological reasoning and combinatorial reasoning about protocols, which has long been a subject of interest. In particular, our topological proof  reasons in a novel way about subcomplexes of the protocol complex, and sheds light on an open question posed by Guerraoui and Pochon (2009).
Finally, using the machinery developed in the design of this unbeatable protocol, we propose a protocol for uniform $k$-set consensus that beats all known solutions by a large margin.

\end{abstract}


\section{Introduction}

The \defemph{k-set consensus} problem~\cite{Chaudhuri90}, which allows processes to decide on up to~$k$ distinct values,
has played an important role in the study of distributed systems for more than two decades.
Its analysis provided deep insights into concurrency and solvability of tasks in fault-tolerant settings, both synchronous and asynchronous
(e.g. \cite{AGGT, BorowskyG93,CHLT, GGP,GHP, Herlihy:1999, SaksZ00}).
Whereas combinatorial techniques sufficed for the study of the traditional (1-set) consensus problem \cite{FLP},  establishing lower bounds and impossibility results for the more general
$k$-set consensus has proven to be more challenging. Tackling these has given rise to the introduction of topological techniques to the theory of distributed systems, and these have become a central tool (see~\cite{Herlihy2013} for a detailed treatment of the subject).

In the synchronous message passing model, the literature distinguishes between uniform and  nonuniform variants of the classic consensus problem.  In nonuniform consensus only the correct processes (the ones that do not crash) are required to decide on the same value, while values decided on by processes who crash are allowed to deviate from the
value decided by
correct processes. In uniform consensus, however, all decisions must be the same. In asynchronous models the two variants coincide, since there is never a guarantee that a silent process has crashed.
Perhaps due to the fact that $k$-set consensus was first studied in the asynchronous setting, only its uniform variant (in which the total number of distinct values decided by correct \emph{and faulty} processes is no greater than~$k$) has been considered in the literature.

This paper is concerned with $k$-set consensus in the synchronous message-passing model. In this setting, there are well-known bounds relating the degree of coordination that can be achieved (captured by the parameter~$k$), the number of process crashes that occur
in a given execution (typically denoted by~$f$), and the time required for decision, which is
$\lfloor \nicefrac{f}{k} \rfloor+2$ in the uniform case (see \cite{AGGT,GGP,GHP}).
Given an {\it a priori} bound of~$\tee$ on the total number of crashes in a run, the worst-case lower bound for decision is $\lfloor \nicefrac{\tee}{k} \rfloor+1$
\cite{CHLT}.

We consider both uniform and nonuniform variants of $k$-set consensus. For the nonuniform case, we present a protocol
that we call
$\OptMink$,
which is \defemph{unbeatable}, in the sense of \cite{AYY-DISC}. Unbeatability is a very strong form of optimality: A $k$-set consensus protocol~$P$ that always decides at least as soon as $\OptMink$ cannot have even one process ever
decide
strictly earlier than
when that
process
decides
in $\OptMink$.
For the uniform case, we present
a protocol that we call
$\UOptMink$,
which is built using methods similar to those used for
$\OptMink$,
and which
strictly
beats
all known protocols, often by a large margin. In many cases, $\UOptMink$ decides in 2 rounds against an adversary for which the best known early-deciding protocols would decide in $\lfloor \nicefrac{\tee}{k} \rfloor+1$ rounds. Whether $\UOptMink$ is unbeatable remains an open problem.

The first unbeatable protocols for consensus (in both uniform and nonuniform variants) were presented in \cite{AYY-DISC}.
The analysis of those unbeatable consensus protocols
shows that the notion of a \defemph{hidden path} is central to the inability to decide in consensus.
Roughly speaking,
a hidden path w.r.t.\  a process~$i$ at the end of round~$m$ is a sequence
of processes that crash one after the other and
could inform
some process~$j$
at time~$m$ of an
initial
value
unknown to
$i$.

Our improved $k$-set consensus protocols are based on an observation that the time required for decision does not depend simply on the number of processes that crash in each round. Only failures that occur in a very specific
pattern, maintaining what we call a \defemph{hidden capacity} of~$k$, can prevent processes from being able to decide.
The hidden capacity is a generalization of the notion of a hidden path, which, as mentioned above,
was shown in \cite{AYY-DISC} to play
an important role in consensus.

The main contributions of this paper are:

\begin{itemize}
\item We provide new solutions to $k$-set
consensus
in the synchronous message passing model with crash failures. For the nonuniform case, we present the first unbeatable protocol for $k$-set consensus, called $\OptMink$: No protocol can
beat the decision times of
$\OptMink$.
For uniform $k$-set consensus, we present a protocol $\UOptMink$ that strictly
beats
all known early-deciding solutions in the literature \cite{AGGT, CHLT, GGP,GHP, RRT}, in some cases beating (all of) them by a large margin.

\item We identify a quantity called the hidden capacity
of
a given execution
w.r.t.\ a process $i$ at time $m$,
which plays a major role in determining the decision times in runs of $k$-set consensus. Roughly speaking, once its hidden capacity drops below~$k$, process~$i$ can decide. Maintaining a hidden capacity of~$k$ requires at least~$k$ processes to crash in every round. However, they must crash in a very particular fashion. In previous solutions to $k$-set consensus in this model (see, e.g., \cite{GGP}),
a process that observes~$k$ or more new failures per round will not decide. In our protocols,
decision can be delayed only as long as failures maintain a hidden capacity of~$k$ or more.
\item Proving unbeatability has the flavor of a lower bound or impossibility proof.
Unsurprisingly, proving
the unbeatability of $\OptMink$
is extremely subtle. We present two different proofs of its unbeatability. One is a completely constructive, combinatorial proof, while the other is a nonconstructive, topological proof based on Sperner's lemma.
The latter is a new style of topological proof, since it addresses the local question of when an individual process can decide, rather than  the more global question of when the last decision is made.
This sheds light on the open problem and challenge posed by Guerraoui and Pochon in~\cite{GP09} regarding how topological reasoning can be used to obtain bounds on local decisions.
Topological lower bounds for $k$-set consensus (e.g. \cite{CHLT,GHP}) have established that $k-1$ connectivity (in addition to a Sperner coloring) precludes the possibility that all processes decide.
Our analysis provides further insight into the topological analysis for local decisions, illustrating  that hidden capacity of~$k$ implies
$k-1$ connectivity of a subcomplex of the protocol complex in a given round.
 The hidden capacity thus explains the source of topological connectivity that underlies the lower-bound proofs for $k$-set consensus.
 To the best of our knowledge, this is the first time knowledge and topological techniques directly interact to obtain a topological characterization.
\end{itemize}

The appendix
contains the detailed technical analysis that supports the claims in the main part of the paper. This includes
full proofs of all technical claims. In particular, both the full combinatorial proof of unbeatability and
the corresponding
topological proof are presented in \cref{sec-proofs-k-set}.

\section{Preliminary Definitions}
\label{sec:model}

\subsection{Computation and Communication Model}
\label{sec:model-communication}

Our model of computation is a synchronous, message-passing model with
benign crash failures.
A system has~\mbox{$n\!\ge\!2$} processes denoted by
$\Proc=\{1,2,\ldots,n\}$.
Each pair of processes is connected by a two-way communication link,
and each message is tagged with the identity of the sender.
Processes
share a discrete global clock that starts
at time~$0$ and
advances by increments of one. Communication in the system proceeds in
a sequence of \defemph{rounds}, with round~$m+1$ taking place between
time~$m$ and time~$m+1$.
Each process starts in some \defemph{initial state} at time~$0$,
usually with an \defemph{input value} of some kind.
In every round, each process first performs a local computation, and performs local actions,
then
it sends a set of messages to other processes, and finally receives messages sent to it
by other processes during the same round.
We consider the
local computations and sending actions of round~$m+1$ as being performed at time~$m$,
and the messages are received at time~$m+1$.

A faulty process fails by \defemph{crashing} in some round~$m\ge 1$.
It behaves correctly in the first~$m-1$ rounds and
sends no messages from round~$m+1$ on.
During its crashing round~$m$, the process may succeed in
sending messages on an arbitrary subset of its links.
At most~$\tee \leq n-1$ processes fail in any given execution.

It is convenient to consider the state and behavior of processes at different (process-time) nodes, where a \defemph{node}  is a pair $\node{i,m}$ referring to process~$i$ at time~$m$.
A \defemph{failure pattern} describes how processes fail in an execution.
It is a layered graph~$\FP$ whose vertices are
all nodes
$\node{i,m}$ for $i\in\Proc$ and $m\ge 0$.
An edge has the form $(\node{i,m-1},\node{j,m})$
and it denotes the fact that a message sent by~$i$ to~$j$ in round~$m$ would be delivered successfully.
Let~$\Crash(\tee)$ denote the set of failure patterns in which
all failures are crash failures, and
at most~$\tee$ crash failures can occur.
An \defemph{input vector} describes what input the processes receive in an
execution. The only inputs we consider are initial values that processes obtain at time~0.
An input vector is thus a tuple $\vec{v}=(v_1,\ldots,v_n)$ where~$v_j$ is the input to process~$j$.
We think of the input vector and the failure pattern as being determined by an external scheduler, and thus a  pair $\alpha=(\vec{v},\FP)$ is called an
\defemph{adversary}.

A \defemph{protocol} describes what messages a process sends and what decisions it takes,
as a deterministic function
 of its local state at the start of the round.  Messages received during a round affect the local state at  the start of the next round.
We assume that a protocol~$P$ has access to the number of processes~$n$ and to the bound~$\tee$,
typically passed to~$P$ as parameters.

A \defemph{run} is a description of an infinite behavior of the system.
Given a run~$r$ and a time~$m$,
the \defemph{local state} of process~$i$ at time~$m$ in~$r$
is denoted by $r_i(m)$,
and the \defemph{global state} at time $m$
is defined to be $r(m)=\node{r_1(m),r_2(m),\ldots,r_n(m)}$.
A protocol~$P$ and an adversary~$\alpha$ uniquely determine a run,
and we write $r = P[\alpha]$.

Since we restrict attention to benign failure models and focus on decision times and solvability in this paper, Coan showed that it is sufficient to consider {\em full-information} protocols ({\em fip}'s for short), defined below \cite{Coan}.
There is a convenient way to consider such protocols in our setting.
With an adversary $\alpha=(\vec{v},\FP)$ we associate a \defemph{communication graph}
$\CG_\alpha$,
consisting of the graph~$\FP$ extended by labeling the initial nodes $\node{j,0}$ with the initial states $v_j$ according to~$\alpha$.
With every node $\node{i,m}$ we associate a subgraph  $\Ga(i,m)$ of~$\CG_\alpha$, which we think of as $i$'s
\defemph{view} at $\node{i,m}$.
Intuitively, this graph
represents
all nodes $\node{j,\ell}$ from which $\node{i,m}$ has heard, and the initial values it has seen.
Formally, $\Ga(i,m)$ is defined by induction on~$m$.
$\Ga(i,0)$ consists of the node $\node{i,0}$, labeled by the initial value~$v_i$.
Assume that $\Ga(1,m),\ldots,\Ga(n,m)$ have been defined, and let $J\subseteq\Proc$ be the set of processes~$j$ such that $j=i$ or $e_j=(\node{j,m},\node{i,m+1})$ is an edge of~$\FP$. Then
$\Ga(i,m+1)$ consists of the node $\node{i,m+1}$, the union of all graphs $\Ga(j,m)$ with $j\in J$, and the edges
$e_j=(\node{j,m},\node{i,m+1})$ for all $j\in J$.
We say that $(j,\ell)$ is
\defemph{seen} by $\node{i,m}$ if $(j,\ell)$ is a node of $\Ga(i,m)$. Note that this occurs exactly
if $\FP$ allows a (Lamport) message chain from $\node{j,\ell}$ to $\node{i,m}$.

A full-information protocol $P$ is one in which at every node $\node{i,m}$ of a run $r=P[\alpha]$ the process~$i$ constructs $\Ga(i,m)$ after receiving its round~$m$ nodes, and sends $\Ga(i,m)$ to all other processes in round~\mbox{$m+1$}. In addition, $P$ specifies what decisions $i$ should take at $\node{i,m}$ based on $\Ga(i,m)$.
Full-information protocols thus differ only in the decisions taken at the nodes.
Finally, in a run $r=P[\alpha]$, we define the local state $r_i(m)$ of a process~$i$ at time~$m$ to be the pair
$\langle \beta,\Ga(i,m)\rangle$, where $\beta=\mbox{`}\bot\mbox{'}$ if $i$ is undecided at time~$m$, and if $\beta=\mv$ in case~$i$ has decided~$\mv$
at or before time $m$.

For ease of exposition and analysis, all of our  protocols are
\fip's.
However, in fact, they can all be implemented in such a way that any process sends any other
process a total of $O(n\log n)$ bits throughout any execution
(see \cref{sec-com-eff}).

\subsection{Domination and Unbeatability}

A protocol~$P$ is a \defemph{worst-case optimal} solution to a decision problem~$S$ in a given model
if it solves~$S$, and decisions in~$P$ are always taken no later than the {\em worst-case} lower bound for decisions in this problem, in a given model of computation.
However, this protocol can be strictly improved upon by \defemph{early stopping} protocols,
which are also worst-case optimal, but can often decide much faster than the original ones.
In this paper, we are interested in protocols that are efficient in a much stronger sense.

Consider a  context $\gamma=(\Vecs,\Fmodel)$, where $\Vecs$ is a set of initial vectors.
A decision protocol $Q$ \defemph{dominates} a protocol~$P$ in~$\gamma$, denoted by $Q\boldsymbol{\dom_\gamma} P$ if, for all adversaries $\alpha$ and every process~$i$,
if $i$ decides in~$P[\alpha]$ at time~$m_i$, then $i$ decides in $Q[\alpha]$ at some time
$m'_i\le m_i$. Moreover, we say that $Q$  \defemph{strictly dominates} $P$
if $Q\dom_\gamma P$ and  $P\!\!\boldsymbol{\not}\!\!\!\dom_\gamma Q$. I.e., if it dominates~$P$ and for some $\alpha\in\gamma$ there exists a process~$i$ that decides in $Q[\alpha]$ {\em strictly before} it does so in $P[\alpha]$.

Following \cite{HMT11},
a protocol~$P$ is said to be an  \defemph{all-case optimal} solution to a decision task~$S$ in a context~$\gamma$ if it solves~$S$ and, moreover, $P$ dominates every protocol~$P'$ that solves~$S$ in~$\gamma$.
For the standard ({\em eventual}) variant of consensus, in which decisions are not required to occur simultaneously, Moses and Tuttle showed that no all-case optimal solution exists~\cite{MT}.
Consequently, Halpern, Moses and Waarts in \cite{HalMoWa2001} initiated the study of a natural notion of optimality
that is achievable by eventual consensus protocols:

\begin{definition}[\cite{HalMoWa2001}]
A protocol $P$ is an \defemph{unbeatable} solution to a decision task~$S$ in a context~$\gamma$ if $P$ solves~$S$ in~$\gamma$ and no protocol $Q$ solving~$S$ in~$\gamma$ strictly dominates~$P$.\footnote{Unbeatable protocols were called {\em optimal} in \cite{HalMoWa2001}. Following \cite{AYY-DISC},
we prefer the term {\em unbeatable} because ``optimal'' is used very broadly, and inconsistently, in the literature.}
\end{definition}

Thus, $P$ is unbeatable if for all protocols~$Q$ that solve~$S$, if there exist an adversary~$\alpha$ and process~$i$ such that~$i$
decides in $Q[\alpha]$ strictly earlier than it does in $P[\alpha]$, then there must exist some adversary~$\beta$ and process~$j$ such that $j$ decides strictly earlier in $P[\beta]$ than it does in $Q[\beta]$. An unbeatable solution for~$S$ is $\dom$-minimal among the solutions of~$S$.

\subsection{Set Consensus}

In the $k$-set consensus problem, each process $i$ starts out with an initial value $v_i\in\{0,1,\ldots,k\}$.\footnote{The set of values in $k$-set consensus is  often assumed to contain more than $k+1$ values. We choose this set for ease of exposition. Our results apply equally well with minor modifications if a larger set of values is assumed. See \cref{more-values}.}
Denote by $\existsv$ the fact that at least one of the processes started out with initial value~$\mathtt{v}$.
In  a protocol for (nonuniform) $k$-set consensus,
the following properties must hold in every run~$r$:

\begin{itemize}
\item[]{\kAgreement:}
The set of values that correct processes decide on has cardinality at most~$k$,
\item[]{\Decision:} Every correct process must decide on some value, and
\item[]{\Validity:} For every value~$\mathtt{v}$, a decision on~$\mathtt{v}$ is allowed only if~$\existsv$ holds.
\end{itemize}

In uniform $k$-set consensus \cite{CBS-uni,Dutta-uni,H86,KR-uni,Raynal04-uni,WTC-uni}, the \kAgreement\ property is replaced by

\begin{itemize}
\item[]{\UnikAg:}
The set of values decided on
has cardinality at most~$k$.
\end{itemize}

\noindent
In uniform $k$-set consensus,
values decided on by failing processes
(before they have failed)
are counted, whereas in the nonuniform case they are not
counted.
The two notions coincide in asynchronous settings, since it is never possible to distinguish at a finite point in time between a crashed process and a very slow one.
Uniformity may be desirable
when elements outside the system can observe decisions, as in distributed databases when decisions correspond to commitments to values.

\section{Unbeatable (1-set) Consensus}
\label{sec:unB-consensus}

Before introducing our protocols for $k$-set consensus, we briefly review the analysis and the unbeatable protocol for nonuniform ($1$-set) consensus
 given in \cite{AYY-DISC}.
In this version of the problem, all processes start with binary initial values $v_i\in\{0,1\}$,
decisions must satisfy the validity condition, and all correct processes decide on the same value. Recall that we are assuming that the processes follow a full-information protocol, so a protocol can be specified simply by giving the rules by which a process decides on value~$\mv$, for $\mv=0,1$.

By Validity, $\existsv$ is a necessary condition for deciding on a value~$\mv\in\{0,1\}$.
Consequently, a process cannot decide~$\mv$ unless it knows that some process had initial value~$\mv$.\footnote{This is an instance of the so-called {\em Knowledge of Preconditions} principle of~\cite{Moses-tark2015}, which states that if ~$\varphi$ is a necessary condition for process~$i$ performing an action~$\sfa$, then $i$ must {\em know}~$\varphi$ when it performs~$\sfa$. For more details on the use of this principle in our setting, see \cref{thm:knowprec} in \cref{sec:know}
and~\cite{AYY-DISC}.}
Clearly, a process knows $\existsv$ if it sees a value of~$\mv$, either as its own initial value or reported in a message that the process receives. \cite{AYY-DISC} consider the design of a protocol that will decide on~0 as soon as possible, namely at the first point at which
a process
 knows~$\existsZ$.
They proceed to consider when a process can decide~1, given that all processes
decide~0 if they ever come to know that~$\existsZ$. Clearly, we should consider the possibility of deciding~1 only for processes that do not know that~$\existsZ$.
By the ($1$-)\Agreement\ property, correct processes must decide on the same value.
Thus, a process cannot decide~1 if another process is deciding~0. Let us consider when a process can know that nobody is deciding~0.

Our analyses will make use of the different types of information that a process~$i$ at time~$m$ can know about the state of a process $j$ at time $\ell$, in runs of an \fip. (We denote such process-time pairs by $\node{i,m}$, $\node{j,\ell}$, etc.) We say that $\node{j,\ell}$ is \defemph{seen by} $\node{i,m}$
if~$i$ has received a message by time~$m$ containing the state at $\node{j,\ell}$. We say that $\node{j,\ell}$ is \defemph{guaranteed crashed} at $\node{i,m}$ if~$i$ has proof at time~$m$ that~$j$ crashed \emph{before} time~$\ell$ ($i$ heard from someone who did not hear from~$j$ in some round~$\le \ell$). Finally, we say that $\node{j,\ell}$ is \defemph{hidden from} $\node{i,m}$ if it is neither seen by $\node{i,m}$ nor guaranteed crashed there.
As far as~$i$ is concerned, $j$ may have sent messages in round~$\ell+1$, and since~$i$ does not see $\node{j,\ell}$, it may not know at~$m$ what information $j$'s messages contained.

\cref{fig-hidden-path} illustrates a case in which process~$i$ does not know $\existsZ$ at time~2, while process~$i_3$ decides~0 at time~2. This is possible only if there is a \defemph{hidden path} with respect to $\node{i,2}$,  in the terminology of~\cite{AYY-DISC}: at each time~$\ell$ from~0 up to the current time $m\!=\!2$, there is a node $\node{j,\ell}$ that is hidden from~$\node{i,2}$.

\begin{figure}[t]
    \centering
        \subfigure[A hidden path w.r.t.\ $\node{i,2}$.]{
            \label{fig-hidden-path:first}
	    \includegraphics[width=0.41\textwidth]{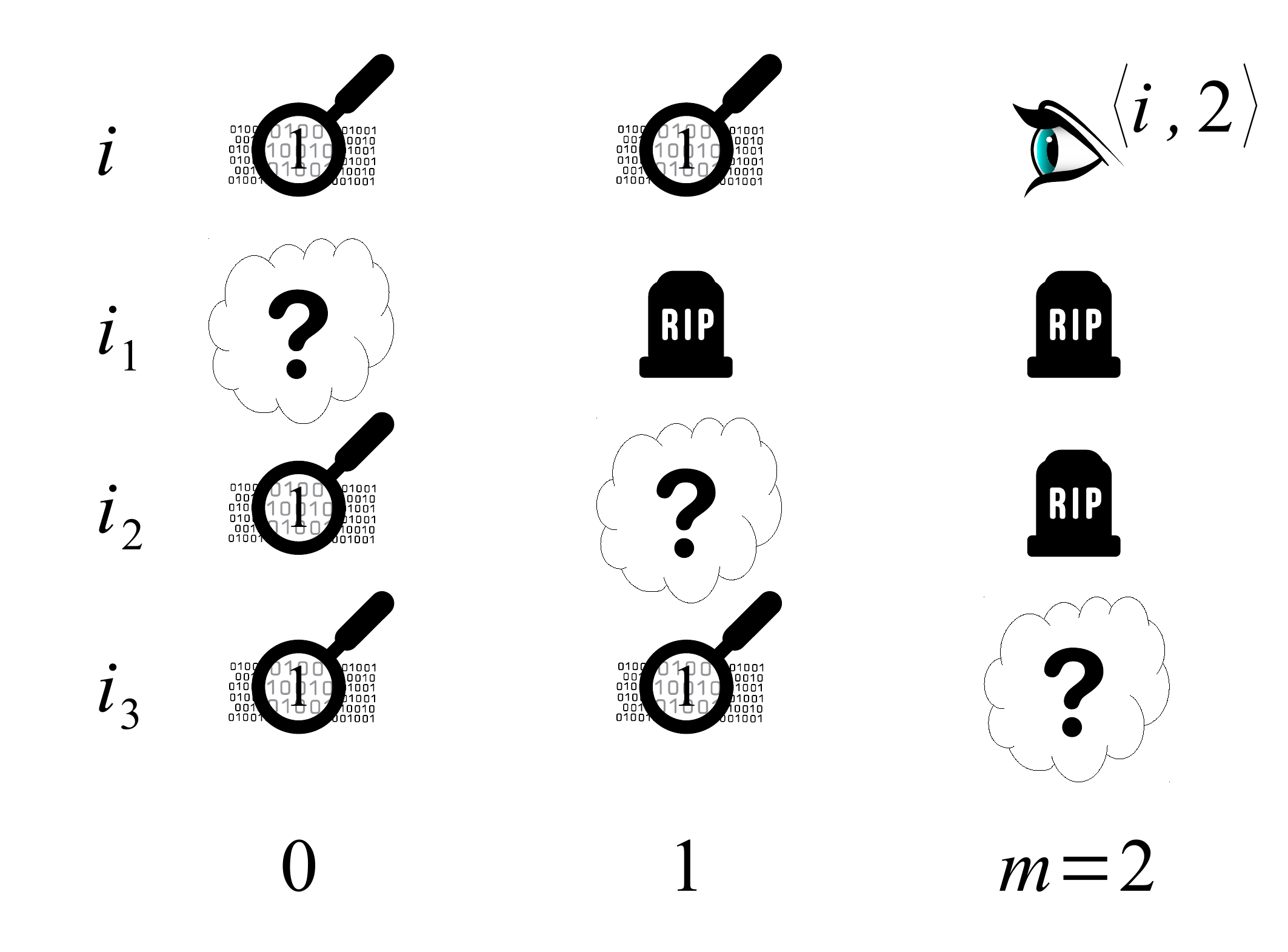}
        }
        \qquad\qquad\qquad
        \subfigure[A run $i$ considers at $2$ to be possible, in which $0$ is held by a correct process at $2$.]{
            \label{fig-hidden-path:second}
	    \includegraphics[width=0.41\textwidth]{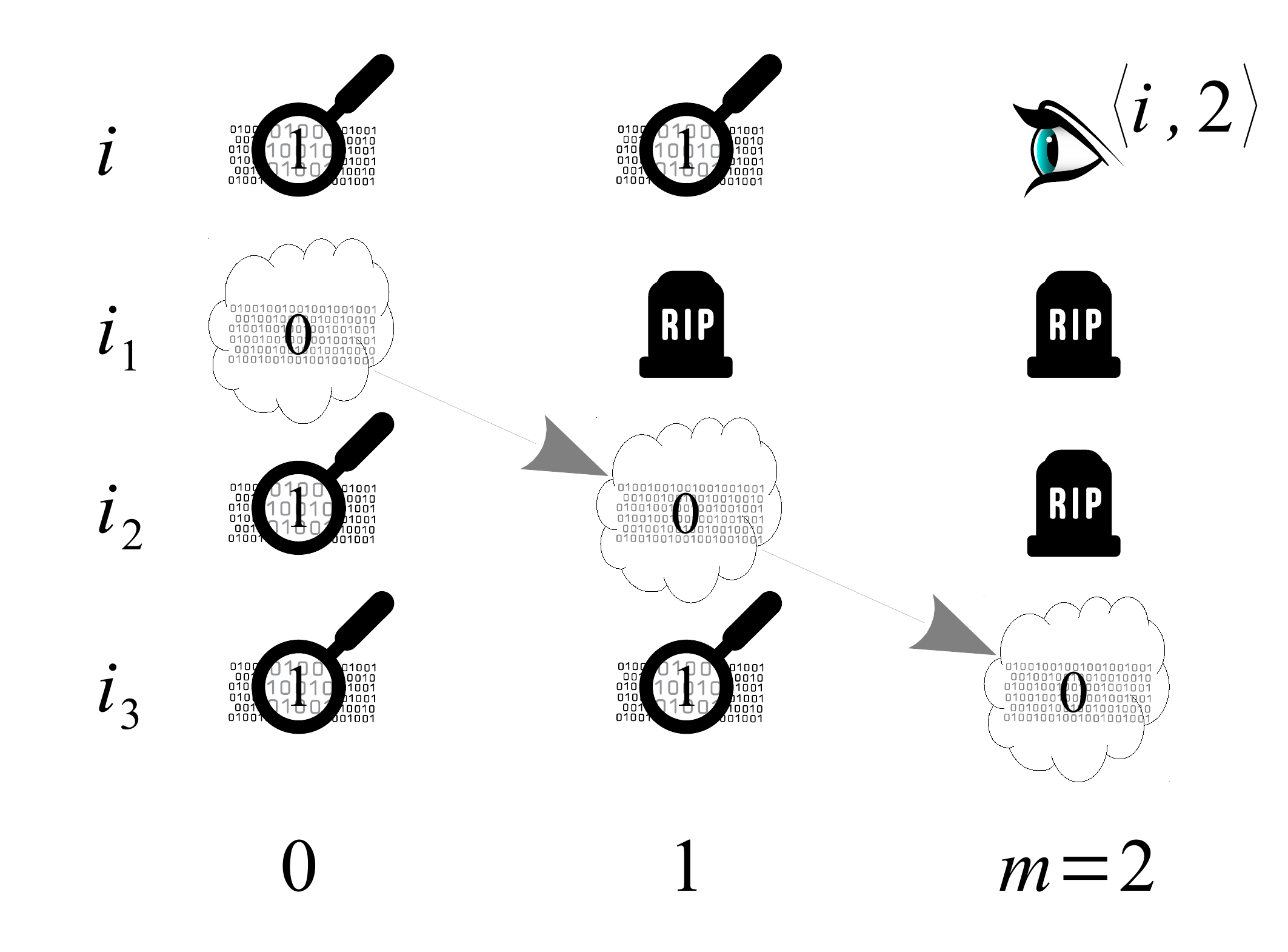}
        }
        \vspace{-0.3cm}
        \caption{
            A hidden path at time $m\!=\!2$ indicates that a value unknown to $i$ may exist in the system.
        }
	\label{fig-hidden-path}
\end{figure}

In general, if there is a hidden path with respect to~$\node{i,m}$ in a given execution, and process~$i$ does not know $\existsZ$ at time~$m$, then~$i$ cannot be guaranteed that no correct process is currently deciding~0. It thus cannot decide~1. If no such path exists, i.e., if there is some time $k\le m$ that contains no hidden node w.r.t.~$\node{i,m}$, then $i$ knows that nobody is deciding~0. Moreover, in that case~$i$ knows that no value of~0 is known to any active process, and so nobody will ever decide~0. Based on this analysis, \cite{AYY-DISC} proposes the following protocol:

\vspace{\topsep}
\noindent
\underline{{\bf Protocol}~$\OptZ$}
 (for an undecided process~$i$ at time~$m$)~\cite{AYY-DISC}:\\[.6ex]
\begin{tabular}{lll}
\quad  {\bf if} & seen~0 & {\bf then}~~$\decideiZ$\\
\quad  {\bf elseif} & some time~$k\le m $ contains no hidden node & \bf then~~$\decideiO$
\end{tabular}
\vspace{\topsep}

As shown in~\cite{AYY-DISC}, $\OptZ$ is an unbeatable protocol for non\-uniform consensus. Moreover, it strictly dominates all previously known early stopping protocols for consensus, in some cases deciding in 3 rounds when the best previously known protocol would decide in~$t+1$ rounds.

Hidden paths,
first defined in \cite{AYY-DISC}, are implicit in many lower-bound proofs for consensus in the crash failure model \cite{DRS,DM}.
 They play a crucial role in the correctness and unbeatability proof of~$\OptZ$.
In this paper, we extend the notion of a hidden path, and use it to provide an unbeatable protocol for set consensus, as well as a protocol for uniform set consensus that beats all previously known protocols.

\section{Unbeatable
{\bf\emph{k}}-Set Consensus}
\label{sec-set-consensus}

The design and especially the analysis of protocols for $k$-set consensus are typically considerably more subtle than that of protocols for ($1$-set) consensus. Indeed, since processes may decide on different values, their decisions no longer depend on hearing about concrete values. Nevertheless, as we now show, it is possible to extend the above notions to obtain a simple protocol for nonuniform $k$-set consensus. While the protocol is natural to derive and simple to state, it is unbeatable. Moreover, its proof of unbeatability is highly nontrivial and extremely subtle.

\subsection{Protocol Description and Correctness}

Recall that the set of possible initial values is assumed to be $\{0,1,\ldots,k\}$.
As in the case of $1$-set consensus, our goal is to define the rules by which a process will decide on~$\mv$, for every value~$\mv$ in this set.
Define by $\minval{i,m}$ the minimal value that process~$i$ has seen by time~$m$
(i.e., the minimal value~$\mv$ s.t.\ $i$ knows that $\exists\mv$)
in a given run of the \fip.
Moreover, we say that process~$i$ is \defemph{low} at time~$m$ if $\minval{i,m}<k$.
We set out to design a protocol in which a process~$i$ that decides at time~$m$  decides on~$\minval{i,m}$.

Suppose that we choose to allow a process that becomes low to immediately decide on~$\minval{i,m}$. This is analogous to allowing a process that sees~0 to immediately decide on~0 in $1$-set consensus. When should a process be able to decide on the value~$k$?
As already specified, every low process decides on a \defemph{low value}, i.e., a value smaller than $k$ as soon as possible; therefore,
we are only concerned about
when high processes (processes that are not low) should decide on a high value, i.e., on $k$.
Of course, a process~$i$ should be able to decide on~$k$ if it knows that doing so will not violate the properties of $k$-set consensus, and in particular, the \kAgreement\ property, which disallows the correct processes to decide on more than~$k$ values. I.e.,
if~$i$ considers it to be possible that there is a correct process deciding on each of the values $0,\ldots,k-1$, then~$i$ should not decide on~$k$. In other words, $i$ can decide on~$k$ only when it knows that at most $k-1$ of the low values will be decided on. Notice that this is not the same as there being only $k-1$ possible low initial values in the run, but rather that at most $k-1$ values can serve as $\minval{j,m}$ for any process~$j$ that
decides
at time~$m$. Consider the following
generalization of the notion of a hidden path:

\begin{definition}
\label{hiddencapacity}
Fix a run~$r$. We define the \mbox{\defemph{hidden capacity}} of process~$i$ at time~$m$, denoted by
$\HC{i,m}$,
to be the maximum
number $c$ such that  for every $\ell\le m$, there exist~$c$ distinct
nodes
$\node{i_1^{\ell},\ell},\ldots,\node{i_c^\ell,\ell}$ at time~$\ell$ that are hidden from $\node{i,m}$.
The nodes $\node{i_b^{\ell},\ell}$ are said to be \defemph{witnesses to the
hidden capacity of}~$i$ at~$m$.
\end{definition}

Put another way, the hidden capacity of $\node{i,m}$ is at least~$c$ iff at each time $\ell\le m$ there are at least~$c$ nodes that are hidden from $\node{i,m}$. In particular, a hidden path implies that the hidden capacity is at least~1. Recall that $\node{i,m}$ does not see what happens at nodes $\node{j,\ell}$ that are hidden from it, and does not know who they heard from in round~$\ell$ and whether they communicated to others in round~$\ell+1$.

As we now show, the hidden capacity is very closely related to the ability of a process to decide~$k$ at a given time. If the hidden capacity of~$\node{i,m}$ is~$c$, then there could be up to~$c$ {\em disjoint}
hidden
paths, where
each failing process that belongs to one of the hidden paths
sends
in its crashing round a message solely to
its successor in
that path. (See \cref{exist-hidden-channels} in \cref{sec-proofs-k-set}.)
\cref{hidden-capacity:first} illustrates the node $\node{i,2}$ with hidden capacity~$3$, while \cref{hidden-capacity:second} illustrates three disjoint hidden paths.
If such disjoint paths begin in nodes with distinct initial values, they can end at~$c$ nodes, each of which sees a distinct minimal value. It follows that a node $\node{i,m}$ with hidden capacity~$k$ must consider it possible that there are~$k$ nodes $\node{j_0,m}$,\ldots,$\node{j_{k-1},m}$ with
$\minval{j_b,m}=b$ for every $b=0,\ldots,k-1$. Therefore, $i$ is not be able to decide on the value~$k$
at time $m$ without risking violating \kAgreement.

\begin{figure}[t]
    \centering
        \subfigure[$\node{i,2}$ has hidden capacity~$3$.]{
            \label{hidden-capacity:first}
	    \includegraphics[width=0.41\textwidth]{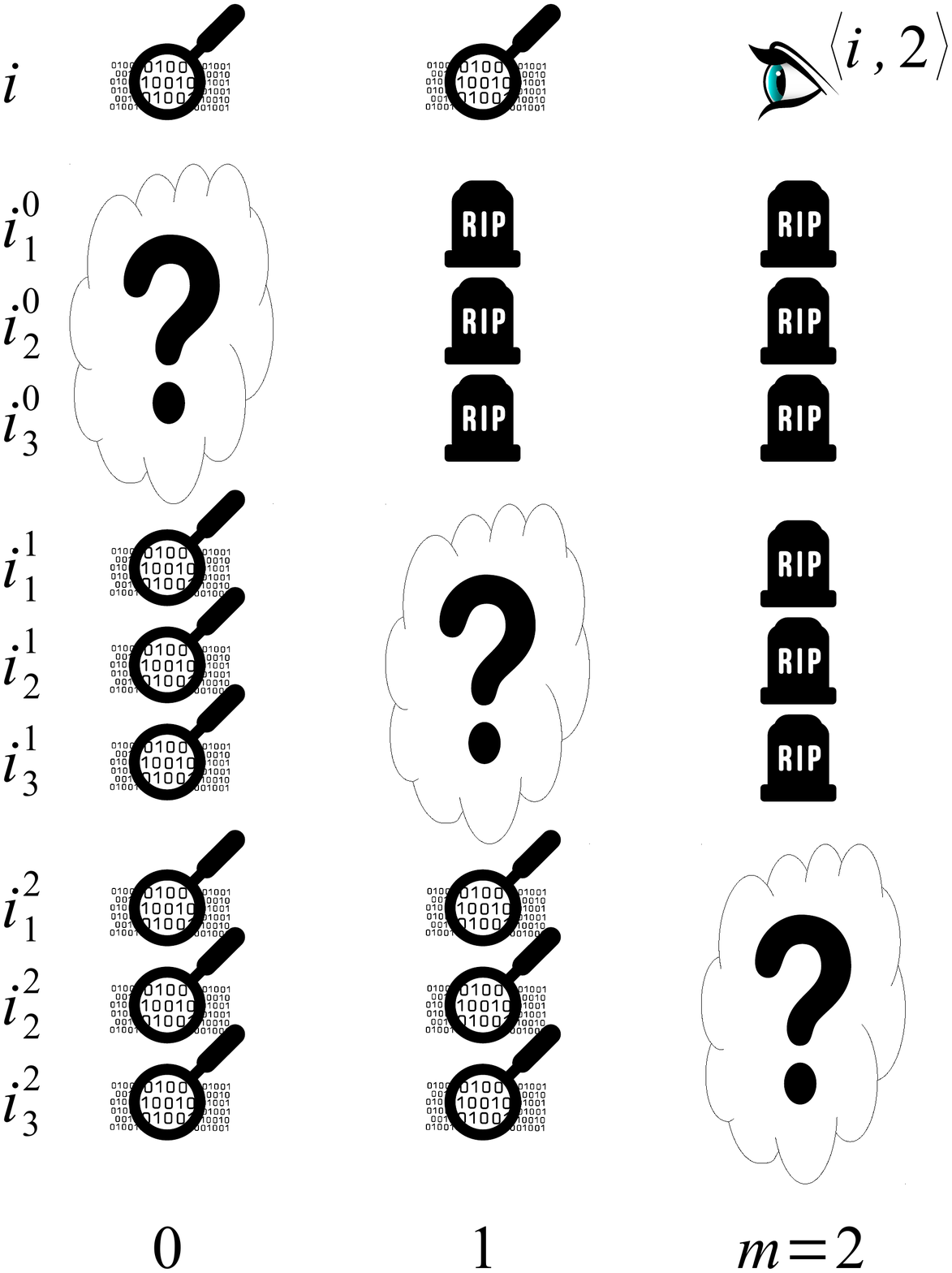}
        }
        \qquad\qquad\qquad
        \subfigure[A run $i$ considers at $2$ to be possible, in which $v_1,v_2,v_3$ are held by distinct processes at $2$.]{
            \label{hidden-capacity:second}
	    \includegraphics[width=0.41\textwidth]{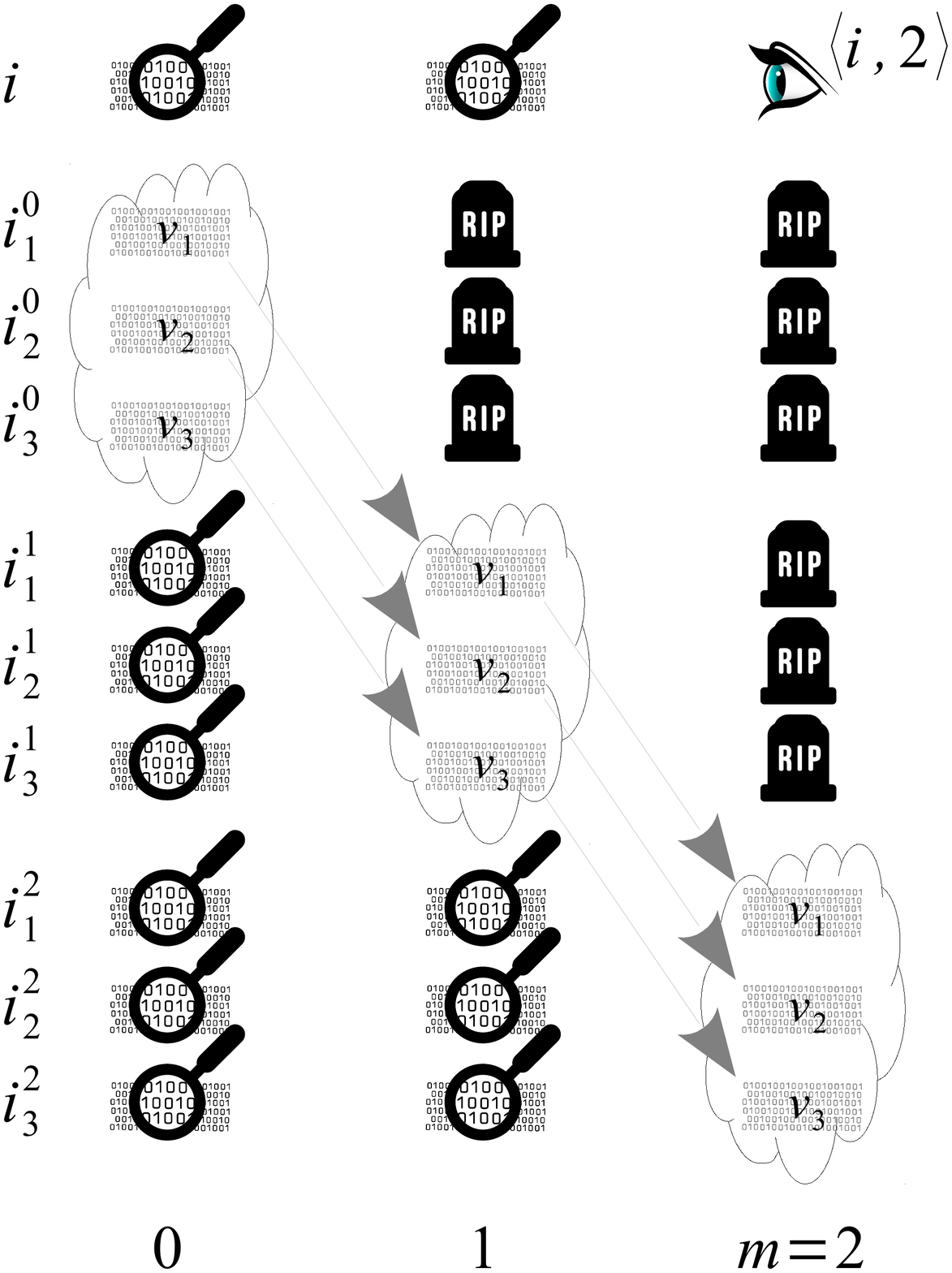}
        }
        \vspace{-0.3cm}
        \caption{
            A hidden capacity of $c\!=\!3$ at time $m\!=\!2$ indicates that any arbitrary $c$ values unknown to $i$ may exist in the system,
            each held by a distinct process (cf.\ \cref{fig-hidden-path}).
        }
	\label{fig-hidden-capacity}
\end{figure}

Interestingly, when the hidden capacity drops below~$k$, it becomes possible to decide on~$k$.
Suppose that the hidden capacity of~$\node{i,m}$ is smaller than~$k$. This means that there must be some time $\ell<m$ such that some $c<k$ nodes at time~$\ell$ are hidden from $\node{i,m}$.
Each of these nodes has a single minimal value. Assuming that
$\minval{i,m}=k$, we have that for time~$\ell$ nodes seen by $\node{i,m}$, the minimal value is~$k$.
Since this is the \fip, every active process after time~$\ell$ will have one of these $c+1\le k$ minimal values. It follows that~$i$ can safely decide on~$k$ at time~$m$ if its hidden capacity is lower
than~$k$.

The above discussion suggests the following protocol for nonuniform $k$-set consensus:

\vspace{\topsep}
\noindent
\underline{{\bf Protocol}~$\OptMink$}
 (for an undecided process~$i$ at time~$m$):\\[.6ex]
\begin{tabular}{lll}
\quad {\bf if} & $i$ is low or~$i$ has hidden capacity $<k$ & {\bf then}~~$\decide(\minval{i,m})$
\end{tabular}
\vspace{\topsep}

Our protocol $\OptMink$ directly generalizes the unbeatable consensus protocol $\OptZ$. Being low in this case corresponds to seeing~0,
while $\HC{i,m}<k=1$ corresponds to there being
no hidden path.
Recall
that the decision rules of $\OptZ$ may be thought of as follows: a process $i$ decides on the value~$0$ as soon as it knows that
some node had initial value~$0$, and it decides on the value~$1$ as soon as it knows that no correct process will ever decide on the value~$0$
(hence \Agreement\ is not violated).
In $\OptMink$, a process $i$ decides on a low value~$\valv$ (i.e., a value $\valv\in\set{0,\ldots,k-1}$)
as soon as it knows that some node had initial
value~$\valv$, and decides on
the high value $\valv = k$ (indeed, for a high node $\node{i,m}$, we have $\minval{i,m}=k$) as soon as it knows that at most $k-1$ values smaller than $\valv$
will ever be decided on by correct processes (thus satisfying \kAgreement).

Based on the above analysis, we obtain:\footnote{\label{more-values}We note that \cref{k-set-correct}, as well as all other the results that we present, including \cref{thm:OptMink,thm:last-decider,u-k-solve} below, as well as all proofs in the
appendix, hold verbatim even if the set of possible initial values is $\{0,\ldots,d\}$, for some $d\ge k$. (In this case, all values in $\{k,k+1,\ldots,d\}$ are considered high.) In particular, the definition for this case of all protocols, including $\OptMink$, is unchanged.}

\begin{proposition}
\label{k-set-correct}
$\OptMink$ solves $k$-set
consensus, and
all processes decide by time
$\lfloor \nicefrac{f}{k} \rfloor+1$.
\end{proposition}

\subsection{Unbeatability}

Proving that $\OptMink$ is unbeatable
is nontrivial.\footnote{In particular, the challenges involved are significantly greater than those in involved in proving that $\OptZ$ is unbeatable for ($1$-set) consensus (as is often the case with problems regarding $k$-set consensus vs.\ their ($1$-set) consensus counterparts).} The main technical challenge along the way is,
roughly speaking, showing that, e.g., in the scenario depicted in \cref{fig-hidden-capacity}, each of the
hidden
processes at time $m\!=\!2$ (i.e., $i^2_b$ for $b=1,2,3$) decides on the
unique low value ($v_b$) known to it, not merely in $\OptMink$, but in \emph{any} protocol $P$ that dominates $\OptMink$. Indeed, in the scenario depicted in \cref{fig-hidden-capacity}, one could imagine a hypothetical protocol in which all active nodes decide at time $2$ on some high value, say the initial value of process~$i$, thus guaranteeing that $i$ will not violate \kAgreement\ by deciding on this value immediately as well.
The claim that such a hypothetical protocol does not exist is made precise in \cref{two-face}
(see also a discussion in \cref{two-approaches}
below).

\begin{lemma}
\label{two-face}
Let $P$ be a protocol solving nonuniform $k$-set consensus.
Assume that in~$P$, every process $i$ that is low at any time~$m$
must decide by time~$m$ at the latest.
Let~$r$ be a run of~$P$,
let $i$ be a process and let~$m$ be a time.
If the following conditions hold in~$r$:
\begin{enumerate}
\item\label{two-face-first-time}
$i$ is low at~$m$ for the first time,
\item
$i$ has seen a single low value
$\valv$
by time~$m$,
\item
$\HC{i,m}\ge k-1$, and
\item
there exist $k$ distinct processes $j_1,\ldots,j_k$
s.t.\ $\node{j_b,m-1}$
is high
and $\node{j_b,m}$ is hidden from $\node{i,m}$,
for all $b=1,\ldots,k$.
\end{enumerate}
then $i$ decides in~$r$ on its unique low value $\valv$ at time~$m$.
\end{lemma}
\noindent
We note that for
$k=3$, the nodes $i^2_b$ in \cref{fig-hidden-capacity} indeed meet the requirements of \cref{two-face} at time $m=2$. (E.g., for $i^2_1$, we may take the processes called $i,i^2_2,i^2_3$ in \cref{fig-hidden-capacity} to serve as the processes $j_1,\ldots,j_k$ in the statement of the \lcnamecref{two-face}.)
Given \cref{two-face}, using reasoning similar to our discussion above regarding the decision rules in $\OptMink$, we conclude that
a high process with hidden capacity at least $k$ (such as process~$i$ at time $m=2$ in \cref{fig-hidden-capacity}, for $k=3$) cannot decide in
any protocol $P$ that dominates $\OptMink$
without risking violating \kAgreement\ (see \cref{k-set-cant-decide-before} in \cref{sec-proofs-k-set}), from which the
unbeatability
of $\OptMink$ follows (indeed, all undecided nodes in $\OptMink$ are high
and have hidden capacity at least $k$,
and are therefore undecided under $P$ as well).

\begin{theorem}
\label{thm:OptMink}
$\OptMink$ is an unbeatable protocol for non\-uniform $k$-set consensus in the crash failure model.
\end{theorem}

We stress that $\OptMink$ is implementable in such a way that each process sends
each other process $O(n\log n)$ bits throughout the run, and each process requires $O(n)$ local steps in every round (see \cref{sec-com-eff}).
Thus, unbeatability for $k$-set consensus is attainable at a modest price.

\subsubsection{Last-decider unbeatability}
In \cite{AYY-DISC} the authors also consider  a variation on the notion of unbeatability, called \defemph{last-decider unbeatability}, which compares runs in terms of the time at which the last correct process decides.
This notion neither implies, nor is implied by, unbeatability as defined above.
Interestingly, $\OptMink$ is unbeatable in this sense as well (see \cref{sec-last-decider}):

\begin{theorem}
\label{thm:last-decider}
$\OptMink$ is last-decider unbeatable for non\-uniform $k$-set consensus in the crash failure model.
\end{theorem}

\subsection{A Constructive Combinatorial Approach vs.\texorpdfstring{\\}{ }A Nonconstructive Topological Approach}
\label{two-approaches}

In \cref{sec-proofs-k-set}, we provide
two proofs for \cref{two-face}. One is combinatorial and completely constructive,
devoid of any topological arguments, while the other
 is nonconstructive and topological, based on Sperner's lemma.
Our topological proof  reasons in a novel way about subcomplexes of the protocol complex.
Both proofs of
\cref{two-face} are by induction.
In both proofs, the induction hypothesis and the (proof of the) base case are the same.
The proofs differ only in
the induction step.

In the terminology of \cref{two-face},
both proofs start by showing that there exists a run that $\node{i,m}$ finds possible in which each of the $k-1$ nodes at
time $m-1$ that are hidden from $\node{i,m}$ holds a distinct low value other than $\valv$; therefore, by the induction hypothesis, had these $k-1$ nodes not failed, they would have each  decided  on its unique low value. Hence, each of the nodes $j_1,\ldots,j_k$ must consider it possible that all of the hidden nodes that it sees from time $m-1$ have actually decided and are correct. The challenge is to show, without any information about $P$ except for the initial assumption that it dominates $\OptMink$, that there must exist a run $r'$ of $P$ that $\node{i,m}$ finds possible (i.e., a way to adjust the messages received by $j_1,\ldots,j_k$ at time $m$) in which $j_1,\ldots,j_k$ collectively decide on all low values (including $\valv$) at time $m$; see \cref{fig-two-face}.\begin{figure}[t]
    \centering
        \subfigure[$r$, as seen by $\node{i,2}$.]{
            \label{fig-two-face:first}
	    \includegraphics[height=5.7cm]{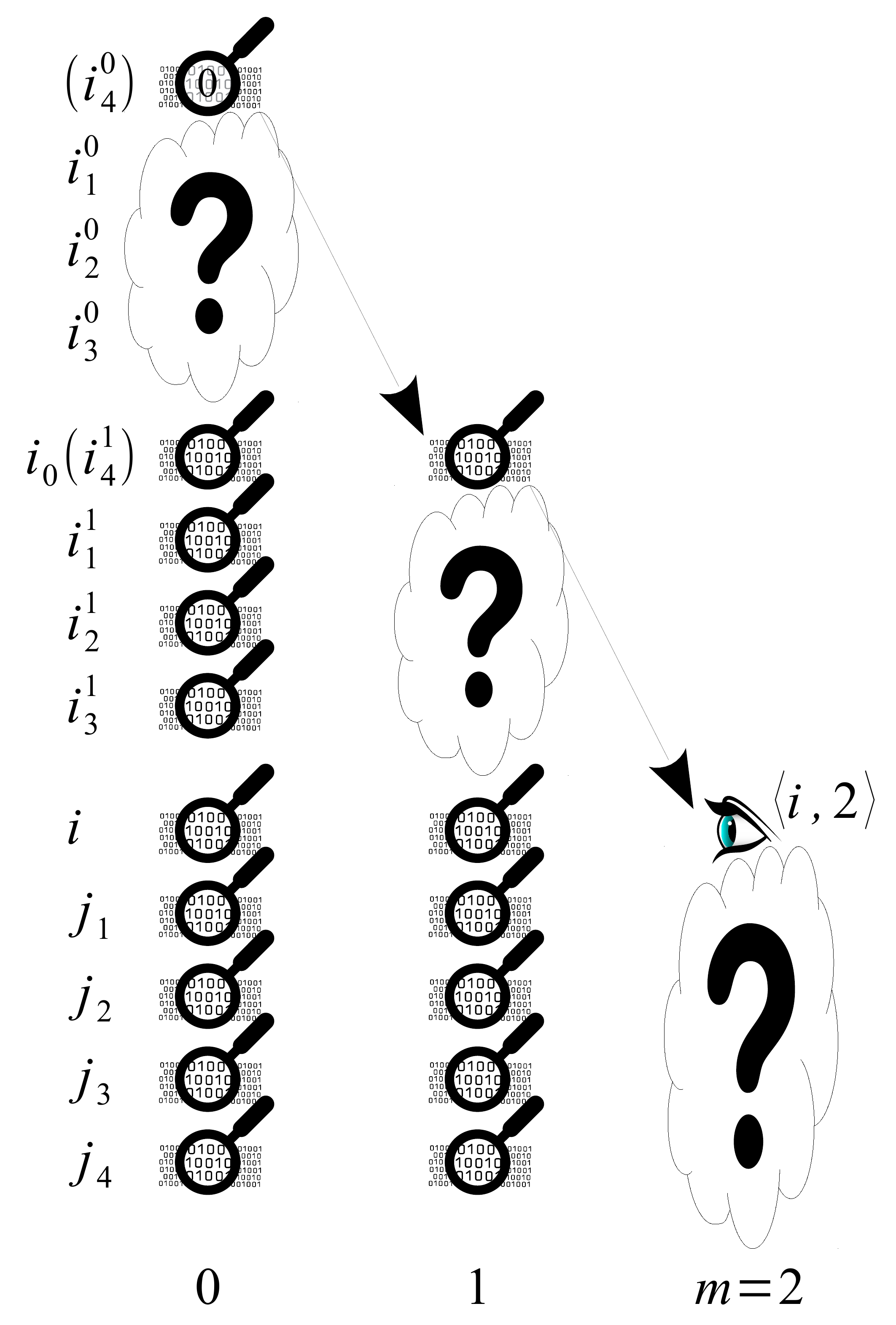}
        }\quad
        \subfigure[The run $r'$.]{
            \label{fig-two-face:second}
	    \includegraphics[height=5.7cm]{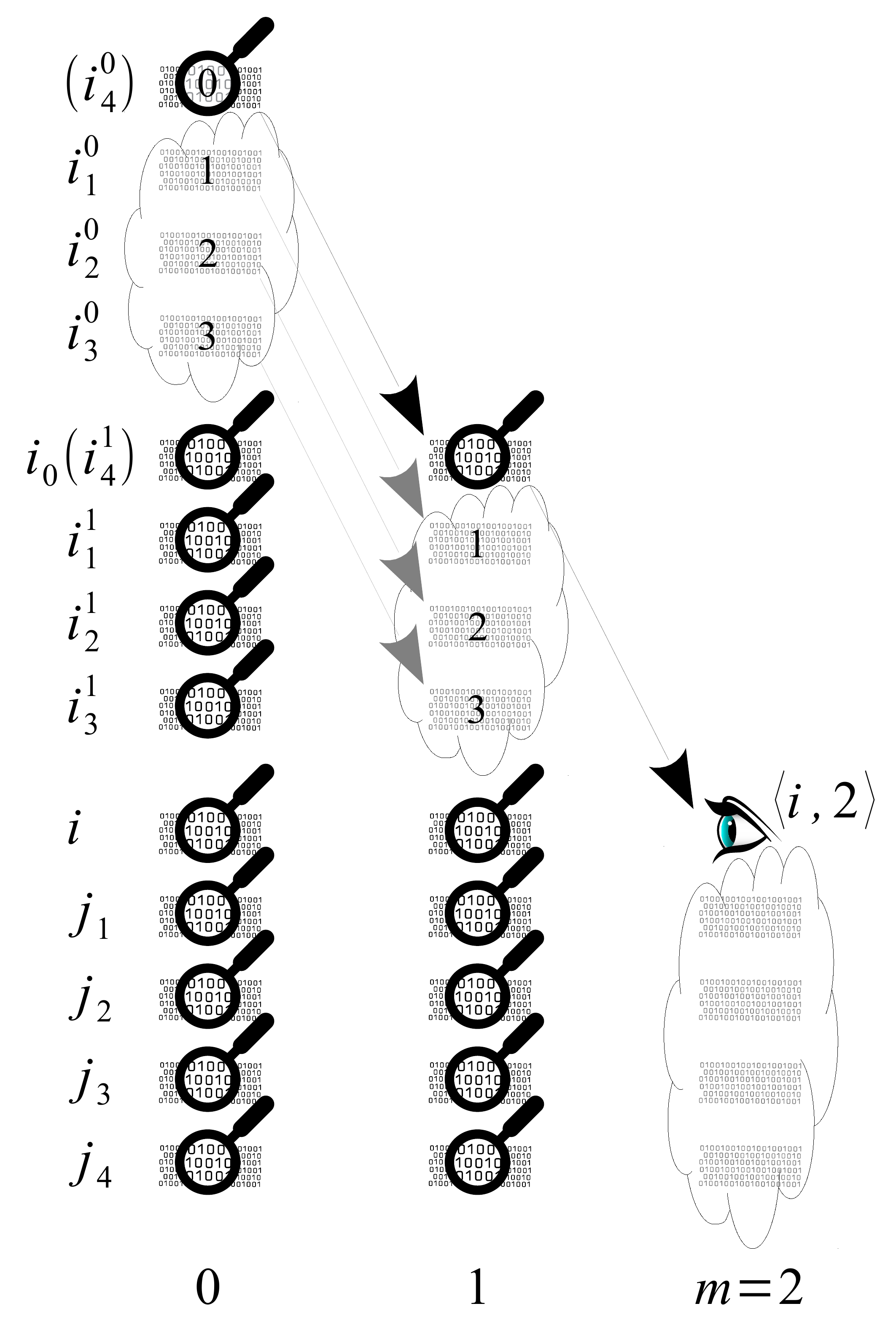}
        }\quad
        \subfigure[$r'$, as seen by $\node{i_3^1,1}$. The induction hypothesis dictates $i_3^1$ decides on $3$ at~$1$.]{
            \label{fig-two-face:third}
	    \includegraphics[trim={0 0 6.5cm 0},clip,height=5.7cm]{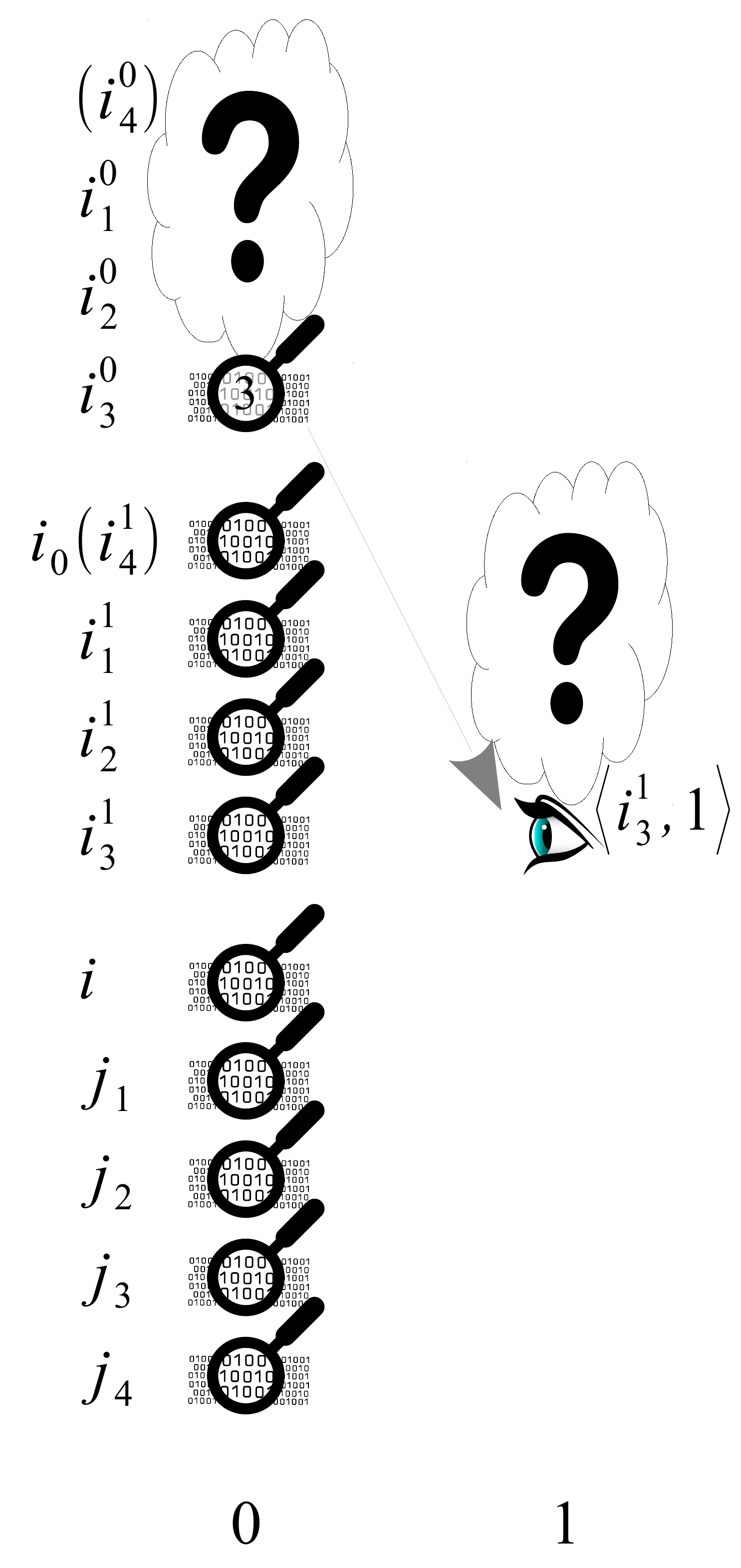}
        }\quad
        \subfigure[We aim to adjust the messages received by $j_1,\ldots,j_4$ at $2$ in $r'$, so that they collectively decide on all low values.]{
            \label{fig-two-face:fourth}
	    \includegraphics[height=5.7cm]{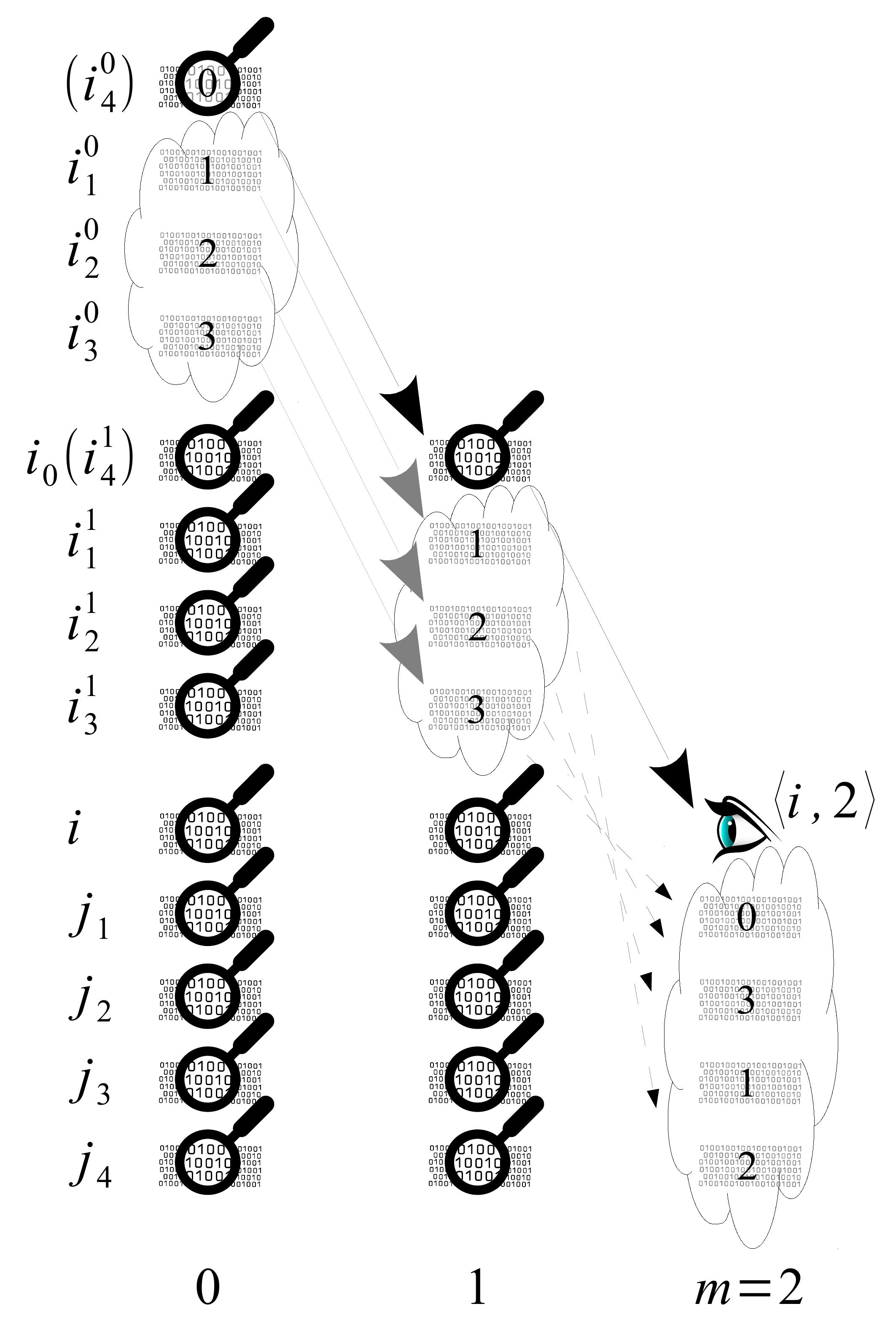}
        }
        \vspace{-0.3cm}
        \caption{
            The proof strategy for the induction step of \cref{two-face} (for $m=2$, $k=4$).
        }
	\label{fig-two-face}
\end{figure}
 Therefore, $i$ cannot decide on a high value without violating \kAgreement\ (and since it must decide, it must do so on a low value, and therefore on the only low value that is has seen, namely $\valv$).

Our combinatorial proof for the induction step constructively builds such a possible run $r'$ as required (providing a ``recipe'' for how to adjust the messages received by $j_1,\ldots,j_k$ at time $m$ so that they collectively decide on all low values).
In contrast, our topological proof uses Sperner's
lemma to show that if process $\ang{i,m}$ does not decide on a low value, then there must exist a run of the protocol that violates \kAgreement\
(i.e., a run $r'$ as described above).
It is interesting that the topological proof essentially shows that
$\ang{i,m}$ is forced to decide on a low value in the run $r$ because the star complex, denoted $\Star(\ang{i,m}, \cP_m)$,
of the node~$\ang{i,m}$ in the protocol complex $\cP_m$ of the protocol~$P$ at time $m$, is $(k-1)$-connected.
Intuitively,
$\Star(v, {\cal P}_m)$
is the ``part'' of $\cP_m$ containing all executions
that are indistinguishable to $\ang{i,m}$.
That $\Star(\ang{i,m}, \cP_m)$ is $(k-1)$-connected is the reason why the proof can map a subdivision of
$k$-simplexes
to process states; indeed, the subdivision is mapped to a subcomplex of $\Star(\ang{i,m}, \cP_m)$.
Therefore, $\ang{i,m}$ has no other choice than to decide on a low value,
because if it does not do so, then its decision induces a Sperner coloring,
which ultimately (together with the connectivity of $\Star(\ang{i,m}, \cP_m)$)  implies that the \kAgreement\ property is violated.

It is worth noticing that in this topological analysis we only
care about the connectivity of a proper subcomplex of the protocol complex,
contrary to all known time-complexity lower-bound proofs~\cite{GHP,HRT98}
for $k$-set consensus, which care about the
connectivity of the whole protocol complex
in a given round.
While connectivity of the whole protocol complex is the ``right'' thing to consider for lower-bound proofs about when
\emph{all} processes can decide,
we show that for proving unbeatability, i.e., when
concerned with the time at which a \emph{single}
process can decide,
the ``right'' thing to consider is the connectivity of just
a subcomplex
(the star complex of a given process state).

This analysis sheds light on the open question posed by Guerraoui and Pochon in~\cite{GP09}
asking for extensions to previous topology techniques that deal with
optimality of protocols.
In summary, while all-decide lower bounds have to do with
the whole protocol complex (e.g.~\cite{HRT98}),
optimal-single-decision lower bounds have to do
with just subcomplexes of the protocol complex. Our topological proof of unbeatability
here is the first proof that we are aware of that makes this distinction.

Finally, we emphasize
that the connectivity properties of the
star
complex $\Star(\ang{i,m}, \cP_m)$
are due to
the hidden capacity of $\ang{i,m}$ in the hypothesis of \cref{two-face}.
Indeed,
one can formally relate
the connectivity of $\Star(\ang{i,m}, \cP_m)$ to the hidden capacity of $\ang{i,m}$,
as
we now show.

\begin{proposition}
\label{hidden-capacity-connectivity}
Let ${\cal P}_m$ be the $m$-round protocol complex
containing all $m$-round executions of
a \fip\ $P$.
If $\ang{i,m}$ is a vertex of ${\cal P}_m$
whose view corresponds to a local state with hidden capacity
at least $k$ in each of the $m$ rounds,
then the star complex $\Star(\ang{i,m}, {\cal P}_m)$ of $\ang{i,m}$ in ${\cal P}_m$ is $(k-1)$-connected.
\end{proposition}

\cref{hidden-capacity-connectivity} speaks about a local property in protocol
complexes, which turns out to be important for optimality analysis.
It is unknown whether the converse of this lemma is true, namely,
whether
$(k-1)$-connectivity of the star complex implies hidden capacity at least $k$ in every round.

\section{Uniform Set Consensus}
\label{subsec-uni-k}

We now turn to consider {\em uniform} $k$-set consensus.
In \cite{AYY-DISC}, the concepts of
hidden paths and hidden nodes
are used to present an unbeatable protocol $\UOptZ$ for ($1$-set) uniform consensus.
In this section, we present a protocol called $\UOptMink$ that generalizes the unbeatable $\UOptZ$ to~$k$ values (i.e.\ for $k=1$, it behaves exactly like $\UOptZ$).
As in the nonuniform case, the analysis of the case $k>1$ is significantly more
subtle and challenging; in fact, generalizing the protocol statement in the uniform case is
considerably
more involved than in the nonuniform case.

While in the protocol $\OptMink$ (which is defined in \cref{sec-set-consensus} for nonuniform consensus) an undecided process~$i$ decides on its minimal value if and only if
$i$ is low or has hidden capacity $<k$, in $\UOptMink$ we have to be more careful. Indeed, we must ensure that a value (even a low one) that process $i$ decides upon  will not ``fade away''.
This could happen if~$i$ is the only one knowing the value,  and if~$i$ crashes without successfully communicating it to active processes.
The case analysis here is also significantly more subtle than in the case of ($1$-set) uniform consensus. To phrase the exact conditions for decision, we begin with a definition; recall that~$\tee$ is
an upper bound on the
number of faulty nodes in any given run, and is available to all processes; while curiously the knowledge of~$\tee$ cannot be used to speed up $\OptMink$, it is indeed useful for speeding up decisions in the uniform case.

\begin{definition}[\cite{AYY-DISC}]
Let $r$ be a run and assume that $i$ knows of $\defemph{d}$ failures at time $m$ in run $r$. We say that $i$  \defemph{knows that the value~$\valv$ will persist} at time $m$ if (at least) one of the following holds.
\begin{itemize}
\item
$m\!>\!0$, and~$i$ both is active at time~$m$ and has seen the value~$\valv$ by time $m-1$, or
\item
$\node{i,m}$ sees at least $(\tee\!-\!\defemph{d})$ distinct nodes $\node{j,m-1}$ of time~$m-1$ that have seen the value~$\valv$.
\end{itemize}
\end{definition}
\noindent
As shown in~\cite{AYY-DISC}, if $i$ knows at time~$m$ that $\mv$ will persist, then all active nodes at time~$m+1$ will know $\exists\mv$.
Everyone's minimal value will be no larger than~$\mv$ from that point on.

In $\UOptMink$,  an undecided process~$i$ decides on a value $\mathtt{v}$ if and only if $\mathtt{v}$ is the minimal value s.t.\ $i$ knows that both
\begin{itemize}
\item
$\mathtt{v}$ was at some stage the min value known to a process that was low or had hidden capacity~$<k$, and
\item
$\mathtt{v}$ will be known to all processes deciding strictly after~$i$.
\end{itemize}
As mentioned above, designing $\UOptMink$ to check that these conditions hold requires
a careful
statement of the protocol, which we now present.

\vspace{\topsep}
\pagebreak
\noindent
\underline{{\bf Protocol}~$\UOptMink$}
 (for an undecided process~$i$ at time~$m$):\\[.6ex]
\begin{tabular}{lll}
\quad{\bf if} & $\hspace{-2em}\bigl(i$ is low or $\HC{i,m}<k\bigr)$ and $i$ knows that $\minval{i,m}$ will persist & {\bf then}~~$\decide(\minval{i,m})$ \\
\quad{\bf elseif} & $m>0$ and $\bigl(\node{i,m-1}$ was low or $\HC{i,m-1}<k\bigr)$ & {\bf then}~~$\decide(\minval{i,m\!-\!1})$ \\
\quad{\bf elseif} & $m=\lfloor \nicefrac{\tee}{k} \rfloor+1$ & {\bf then}~~$\decide(\minval{i,m})$
\end{tabular}
\vspace{\topsep}

The correctness and worst-case complexity of $\UOptMink$ are stated in \cref{u-k-solve}. The reader is referred to its proof in \cref{sec-uni-set-cons-proofs}
for a precise analysis of the decision conditions. We remark  that, roughly speaking, the second condition decides on $\minval{i,m-1}$ and not on $\minval{i,m}$ because the latter value is not guaranteed to persist, while the former value is.

\begin{theorem}
\label{u-k-solve}
$\UOptMink$ solves \defemph{uniform} $k$-set consensus in the crash failure model, and all processes decide by time $\min\bigl\{\lfloor \nicefrac{\tee}{k} \rfloor+1,\lfloor \nicefrac{f}{k} \rfloor +2\bigr\}$.
\end{theorem}

As shown by \cref{u-k-solve}, the protocol $\UOptMink$ meets the worst-case lower bound
 for uniform $k$-set consensus from~\cite{GHP,AGGT}.
 We emphasize that $\UOptMink$ strictly dominates all existing
uniform $k$-set consensus protocols in the literature~\cite{CHLT,GGP,GP09,RRT}.
Essentially, in each of these protocols, a process remains undecided as long as it discovers at least $k$ new failures in every round.
The fact that $\UOptMink$ is based on hidden capacity and hidden paths, rather than on the number of failures seen,
allows runs with much faster stopping times.
In particular, there exist runs in which all previous protocols decide after $\lfloor\nicefrac{\tee}{k}\rfloor+1$ rounds, and in $\UOptMink$ all processes decide by time~$2$; see~\cref{fig-large-margin} for an example.

\begin{figure}[t]
    \centering
        \subfigure[The run $r$. Messages successfully sent by crashing nodes are marked with an arrow. Starting from time $m=1$, all processes know $\exists3$ (i.e., $\exists k$), and so only knowledge of additional values is indicated.]{
            \label{fig-large-margin:first}
	    \includegraphics[height=6.6cm]{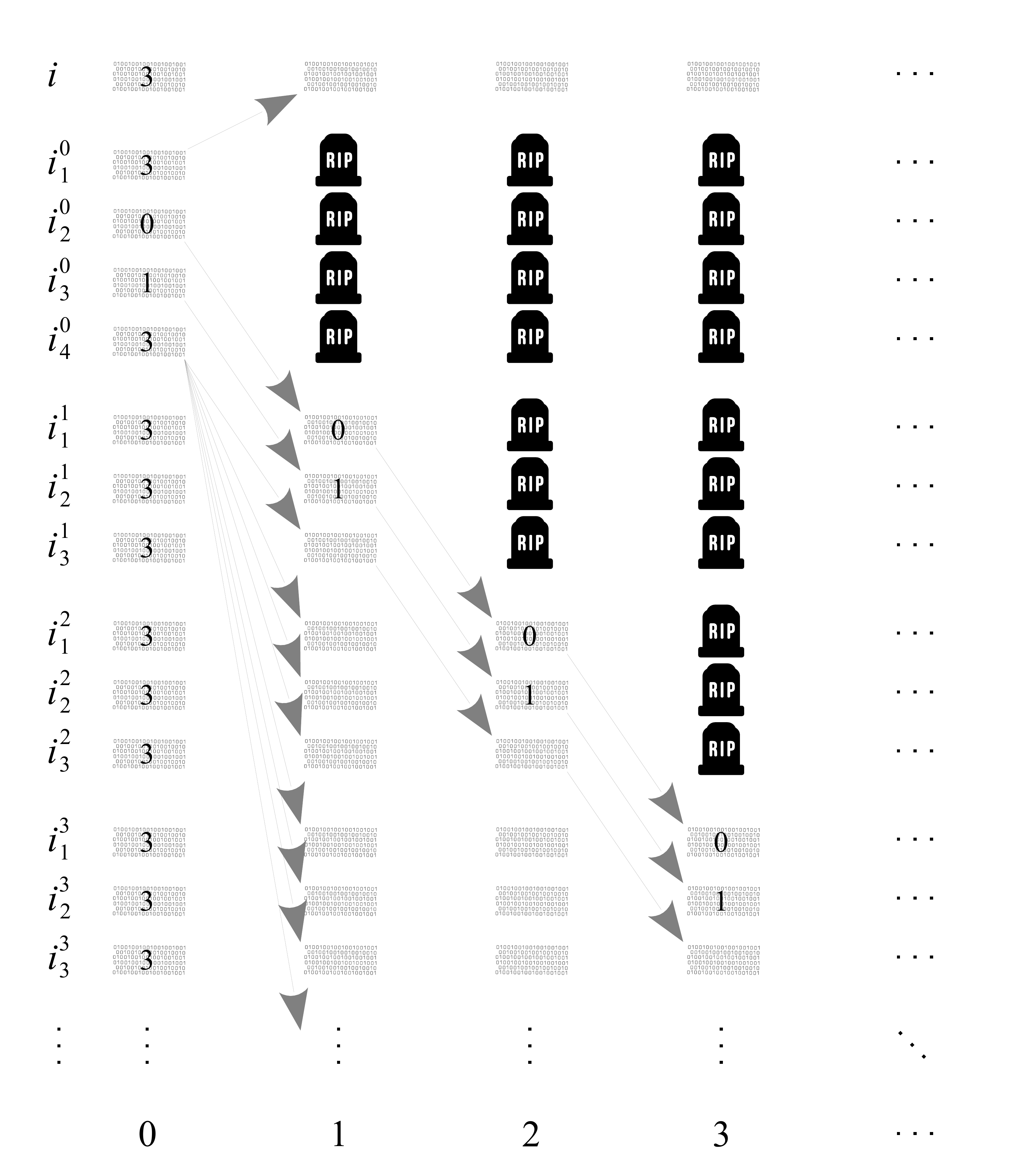}
        }\quad
        \subfigure[$r$, as seen by $\node{i,1}$, which has hidden capacity $k=3$.]{
            \label{fig-large-margin:second}
	    \includegraphics[trim={0 0 19.5cm 0},clip,height=6.6cm]{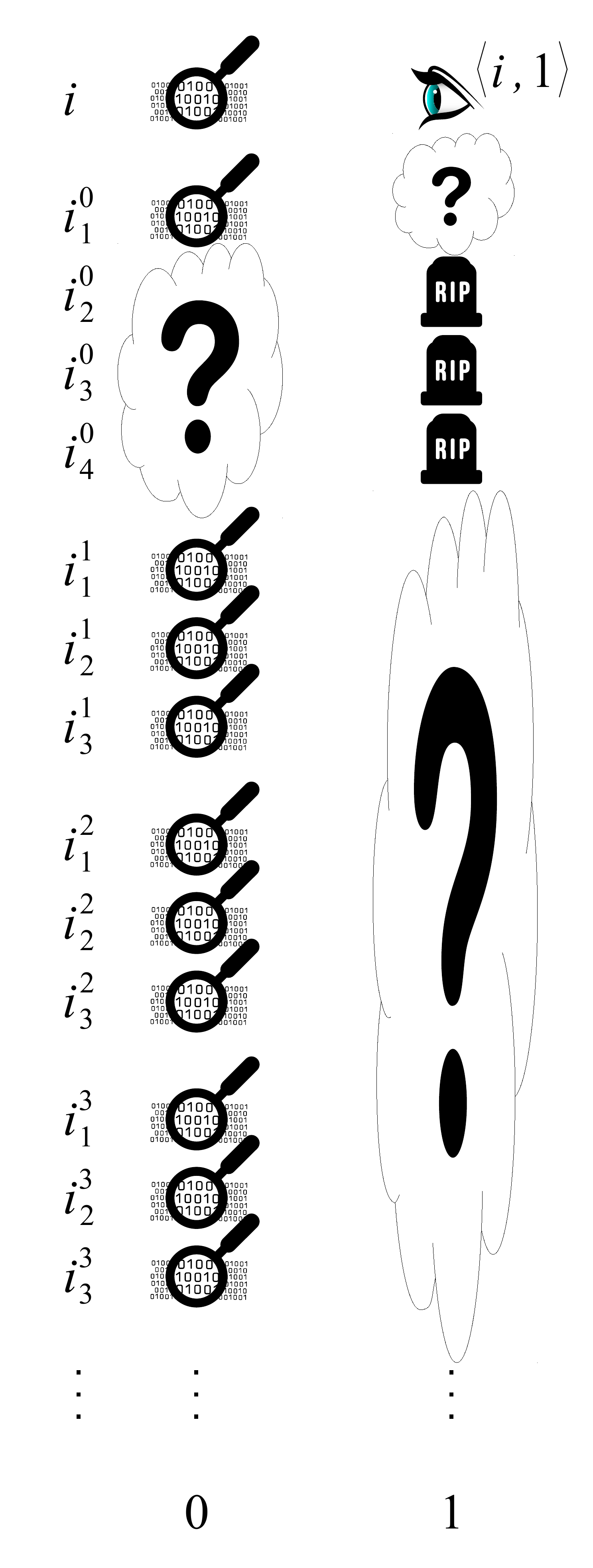}
        }\quad
        \subfigure[$r$, as seen by all other nodes that are nonfaulty at time~$1$. These nodes also have hidden capacity $k=3$. (These nodes are hidden from one another).]{
            \label{fig-large-margin:third}
	    \includegraphics[trim={0 0 19.5cm 0},clip,height=6.6cm]{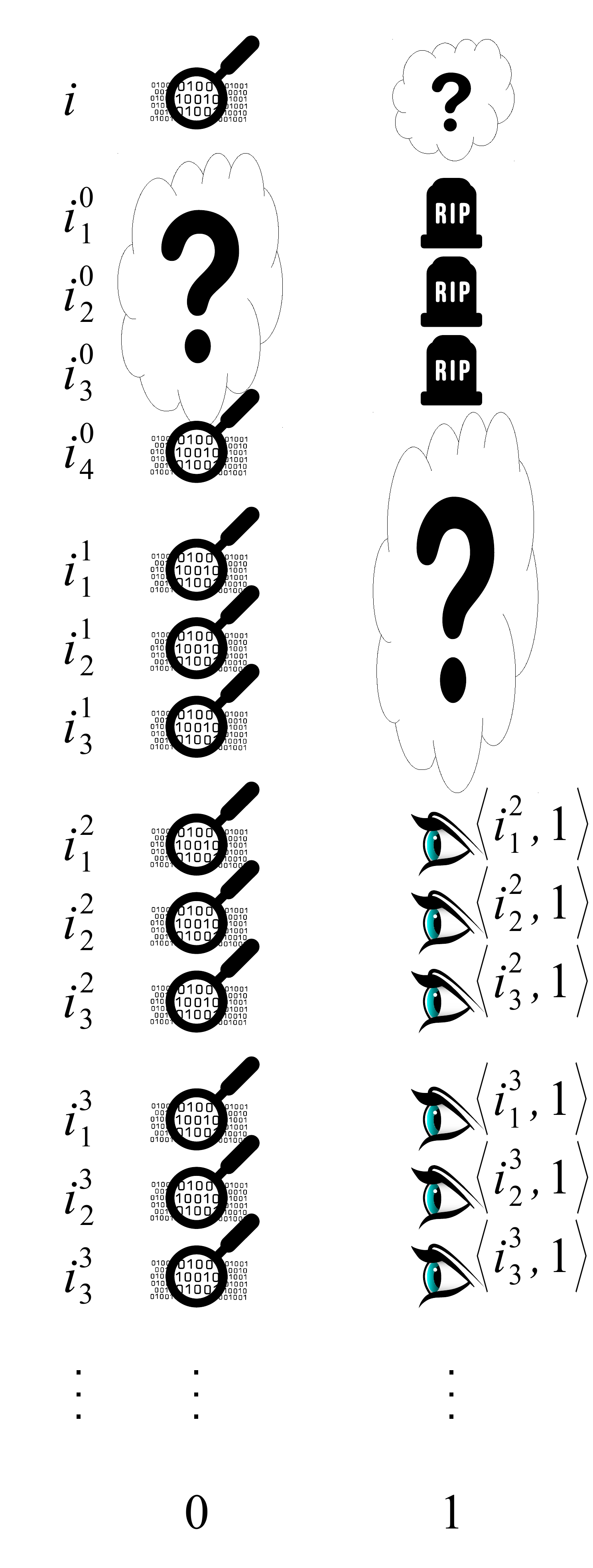}
        }\quad
        \subfigure[$r$, as seen by all nodes that are nonfaulty at time~$2$. All such nodes have hidden capacity $2<k$, and furthermore know that their min value (i.e., $k$) will persist; thus, they all decide (on $k$).]{
            \label{fig-large-margin:fourth}
	    \includegraphics[trim={0 0 13cm 0},clip,height=6.6cm]{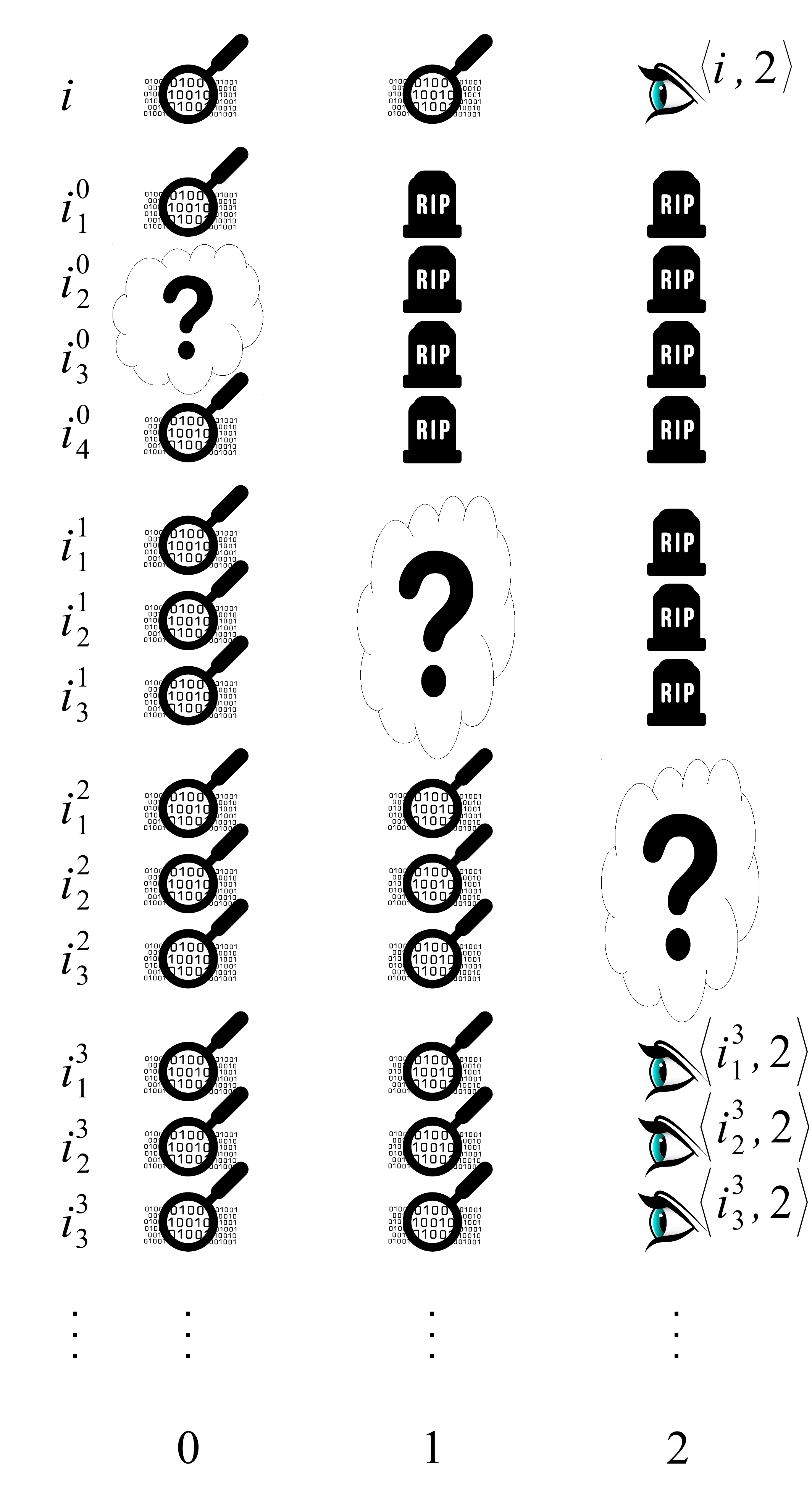}
        }
        \vspace{-0.3cm}
        \caption{
            A run (adversary) s.t.\ in $\UOptMink$ all nodes decide by time~$2$. Nonetheless, in this run, in all previously known protocols (particularly, in the ones of  \cite{CHLT,GGP,GP09,RRT}), the correct processes only at time~\mbox{$\lfloor\nicefrac{\tee}{k}\rfloor+1$} (assuming $f=\tee$), as every correct process sees $k$ new failures in each round for the first $\lfloor\nicefrac{\tee}{k}\rfloor$ rounds. (The \lcnamecref{fig-large-margin} illustrates the case of $k=3$ and arbitrarily high $\tee$.)
        }
	\label{fig-large-margin}
\end{figure}

At this point, however, we have been unable to resolve the following.

\begin{conjecture}
$\UOptMink$~is an \pdo\  uniform $k$-set consensus protocol in the crash failure model.
\end{conjecture}

\vspace{-.7em}
\enlargethispage{.2em}

\section{Discussion}
\label{sec-discussion}

In this paper we have presented two main algorithmic results. The first one, $\OptMink$, is an unbeatable protocol for nonuniform $k$-set
consensus; this protocol has an extremely concise description.
Unbeatability \cite{HalMoWa2001,AYY-DISC} is an optimality criterion that formalizes the intuition that a given protocol cannot be improved upon;\footnote{We emphasize that the notion of ``improvement" captured by the notion of unbeatability studied in this paper is defined in terms of the times at which processes perform their decisions. This is distinct from their halting times, for example (although a process can safely halt at most one round after it decides). Optimizing decision times can come at a cost in communication, for example. Finally, of course, one can compare protocols in terms of more global properties such as average decision times (w.r.t.~appropriate distributions etc.). Another point to notice when considering unbeatability is that it is based on comparing the performance of different protocols on the same behaviors of the adversary. While we find this a reasonable thing to do in benign failure models such as crash and omission failures, it may be rather tricky in the presence of malicious (Byzantine) failures.
}
 this is significantly stronger than saying that a protocol is worst-case optimal.
Our second result is a protocol, $\UOptMink$, for uniform $k$-set consensus that strictly
beats
all known
protocols in the literature~\cite{CHLT,GGP,GP09,RRT};
notably,
in some executions, processes in our protocol can decide
much faster than in those protocols.
Whether our uniform $k$-set consensus protocol is unbeatable remains an open problem.
Both protocols are
efficiently implementable.

We have presented two distinct
proofs for
the unbeatability of our
nonuniform $k$-set consensus
protocol $\OptMink$.
Each proof gives a different perspective
of the unbeatability of the protocol.
The first proof is fully constructive and combinatorial, while the second relies on Sperner's lemma and is nonconstructive and topological.
The topological proof of \cref{two-face} is more than just a ``trick''
to prove the lemma.
In a precise sense,
the proof shows what a topological
analysis of unbeatable protocols is about.

The construction and analysis of both protocols, as well as the unbeatability of $\OptMink$, crucially depends on a new notion hidden capacity.
This notion
is a generalization of the notion of hidden path
introduced in~\cite{AYY-DISC},
which plays a similar role in the case of ($1$-set) consensus.
Derived from our topological unbeatability proof, we identified a connection between hidden capacity $k$ and
local $(k-1)$-connectivity of star subcomplexes of a protocol complex: roughly speaking, hidden capacity at least $k$ implies $(k-1)$-connectivity.
Whether the converse is true remains an open problem.
This connection sheds light on the open question in~\cite{GP09}
asking for extensions to previous topology techniques that deal with
optimality of protocols.
The full interplay between hidden capacity and topological reasoning (even beyond set consensus) is an interesting direction for future research.

\section{Acknowledgments}

\cref{fig-hidden-path,fig-hidden-capacity,fig-two-face,fig-large-margin} include visual elements designed by Freepik.

Armando Casta{\~n}eda is supported partially by a project PAPIIT-UNAM IA101015.
This research was partially done while Armando Casta{\~n}eda was at
the Department of Computer Science of the Technion,
supported by an Aly Kaufman post-doctoral fellowship.

Yannai Gonczarowski is supported by the Adams Fellowship Program of the Israel Academy of Sciences and Humanities; his work is supported by the European Research Council under the European Community's Seventh Framework Programme (FP7/2007-2013) / ERC grant agreement no.\ [249159], by ISF grants 230/10 and 1435/14 administered by the Israeli Academy of Sciences, and by Israel-USA Bi-national Science Foundation (BSF) grant number 2014389.

Yoram Moses is the Israel Pollak chair at the Technion; his work was supported in part by ISF grant 1520/11 administered by the Israeli Academy of Sciences.

\bibliographystyle{abbrv}
\bibliography{z}

\appendix
\section{Knowledge}
\label{sec:know}

Our construction of unbeatable protocols is assisted and guided by a knowledge-based analysis, in the spirit of \cite{FHMV,HM1}. We now define only what is needed for the purposes of this paper, as the proofs in the appendix make formal use of knowledge
. For a comprehensive treatment, the reader is referred to~\cite{FHMV}.
Runs are dynamic objects, changing from one time point to the next. E.g., at one point process~$i$ may be undecided, while at the next it may decide on a value. Similarly, the set of initial values that~$i$ knows about, or has seen, may change over time. In addition, whether a process knows something at a given point can depend on what is true in other runs in which the process has the same information.
We will therefore consider the truth of facts at \defemph{points} $(r,m)$---time~$m$ in run~$r$, with respect to a set of runs~$R$ (which we call a \defemph{system}).
The systems we will be interested will have the form $R_P=R(P,\gamma)$ where $P$ is a protocol and $\gamma=\gamma(\Vals^n,\Fmodel)$ is the set of all adversaries that assign initial values from~$\Vals$ and failures according to~$\Fmodel$. We will write $(R,r,m)\sat A$ to state that fact~$A$ holds, or is satisfied, at $(r,m)$ in the system~$R$.

The truth of some facts can be defined directly.
For example, the fact $\existsv$ will hold at $(r,m)$ in~$R$
if some process had initial value~$\mathtt{v}$ in~$r$. We say that {\em (satisfaction of)} a fact~$A$ is \defemph{well defined} in~$R$ if
for every point $(r,m)$ with $r\in R$ we can determine whether or not $(R,r,m)\sat A$.
Satisfaction of~$\existsv$ is thus well defined. We will write $K_iA$ to denote that \defemph{process~$i$ knows~$A$}, and define:
\begin{definition}[Knowledge]
\label{def:know}
Suppose that~$A$ is well defined in~$R$. Then:

\begin{tabular}{r l c l}
$(R,r,m)$&\hskip-3mm$\sat K_iA$ & iff & $(R,r',m)\sat A$ ~~\mbox{for all~~} $r'\in R$~~\mbox{such that~~}
$r_i(m)=r'_i(m)$.
\end{tabular}
\end{definition}

Thus, if $A$ is well defined in~$R$ then \cref{def:know} makes $K_iA$ well defined in~$R$.
Note that what a process knows or does not know depends on its local state.
The definition can then be applied recursively, to define the truth of $K_jK_iA$ etc. Moreover, any boolean combination of well-defined facts is also well-defined. Knowledge has been used to study a variety of problems in distributed computing.
We will make use of the following  fundamental connection between knowledge and actions in distributed systems. We say that a fact~$A$ is a \defemph{necessary condition} for process~$i$ performing action~$\sigma$
(e.g., deciding
on an output value)
in~$R$ if
$(R,r,m)\sat A$ whenever $i$ performs $\sigma$ at a point $(r,m)$ of~$R$.
\begin{theorem}[Knowledge of Preconditions, \cite{Moses-tark2015}]
\label{thm:knowprec}
Assume that $R_P=R(P,\gamma)$ is the set of runs of a deterministic protocol~$P$.
If $A$ is a necessary condition for~$i$ performing~$\sigma$ in~$R_P$, then $K_iA$ is a necessary condition for~$i$ performing~$\sigma$ in~$R_P$.
\end{theorem}

\section{Nonuniform Set Consensus}
\label{sec-proofs-k-set}

\begin{definition}
Let $r$ be a run,
let $i$ be a process and let $m$ be a time.
We define the following notations, in which
$r$ is implicit.
\begin{enumerate}
\item $\knownvals{i,m}~\eqdef~\{\mathtt{v}: K_i \existsv~\mbox{holds at time~$m$}\}$,
\item $\knownlows{i,m}~\eqdef~ \knownvals{i,m} \cap \set{0,\ldots,k-1}$.
\end{enumerate}
\end{definition}

\begin{remark}
Let $i$ be a node in a run $r$.
\begin{itemize}
\item
The hidden capacity of $i$ in $r$ is (weakly) decreasing as a function of time.
\item
Process~$i$ is low at time~$m$ iff $\knownlows{i,m} \ne \emptyset$.
\item
For all times $m<0$, by definition $\knownvals{i,m}=\emptyset$, and thus $i$ is high.
\item
$\minval{i,m}~\eqdef~\min\knownvals{i,m}$ for every time $m$.
\end{itemize}
\end{remark}

\begin{proof}[Proof of \cref{k-set-correct}]
In some run of $\OptMink$, let $i$ be a nonfaulty process.

\Decision:
Let $m$ be a time s.t.\ $i$ has not decided until $m$, inclusive.
Thus, $\node{i,m}$ has hidden capacity $\ge k$.
Let $i_b^{\ell}$, for all $\ell\le m$ and $b=1,\ldots,k$,
be as in \cref{hiddencapacity}.
By definition, $i_b^{\ell}$, for every $\ell < m$ and $b=1,\ldots,k$,
fails at time $\ell$. Thus, $k\cdot m \le f$, where $f$ is the number of failures
in the current run. Thus, $m \le \nicefrac{f}{k}$, and therefore
$m \le \lfloor \nicefrac{f}{k} \rfloor$.
Therefore, $i$ decides by time $\lfloor \nicefrac{f}{k} \rfloor +1$ at the latest.

Henceforth, let $m$ be the decision time of $i$
and let $\mathtt{v} = \minval{i,m}$ be the value upon which $i$ decides.

\Validity:
As $v = \minval{i,m}$, we have $\mathtt{v} \in \knownvals{i,m}$ and thus
$K_i \existsv$ at $m$. Thus, $\existsv$.

\kAgreement:
It is enough to show that at most $k-1$ distinct values
smaller than $\mathtt{v}$ are decided upon in the current run.
Since $i$ decides at $m$, $\node{i,m}$ is either low or has hidden capacity $<k$.
If $\node{i,m}$ is low, then $\mathtt{v} = \minval{i,m} \le k-1$, and thus
there do not exist more than $k-1$ distinct legal values smaller than $\mathtt{v}$,
let alone ones decided
upon.

For the rest of this proof we assume, therefore, that $\node{i,m}$ is
high and has hidden capacity $<k$.
As $\node{i,m}$ does not have hidden capacity $k$, there exists $0\le\ell\le m$
s.t.\ no more than $k-1$ processes at time $\ell$ are hidden from $\node{i,m}$.

Let $\mathtt{w}<\mathtt{v}$ be a value decided upon by a nonfaulty processor.
Let $j$ be this processor, and let~$m'$ be the time at which $j$ decides on $\mathtt{w}$.
As $\mathtt{w}<\mathtt{v}$ and as $\mathtt{v} = \minval{i,m}$, $\node{j,m'}$ is not seen
by $\node{i,m}$.
As $j$ and $i$ are both nonfaulty, we conclude that $m' \ge m$, and thus
$m' \ge \ell$.
Let $H$ be the set of all processes seen at $\ell$ by $\node{j,m'}$.
Since $m' \ge \ell$,
We have $\knownvals{j,m'} = \bigcup_{h \in H} \knownvals{h,\ell}$. (Note
that if $m' = \ell$, then $H = \set{j}$.)
As $\mathtt{w} = \minval{j,m'}$, we have $\mathtt{w} = \minval{h,\ell}$ for
some $h \in H$. As $\mathtt{w} < \mathtt{v} = \minval{i,m}$, we have
$w \notin \knownvals{i,m}$, and thus $\node{h,\ell}$ is not seen
by $\node{i,m}$. As $\node{h,\ell}$ is seen by $\node{j,m'}$, $h$ has not failed
before $\ell$, and thus $\node{h,\ell}$ is hidden from $\node{i,m}$.
To conclude, we have shown that
\[\mathtt{w} \in \bigl\{ \minval{h,\ell} \mid
\mbox{$\node{h,\ell}$ is hidden from $\node{i,m}$} \bigr\}.\]
As there are at most $k-1$ processes hidden at $\ell$ from $\node{i,m}$,
we conclude that no more than $k-1$ distinct values lower than $\mathtt{v}$ are
decided upon by nonfaulty processes, and the proof is complete.
\end{proof}

\begin{lemma}[See \cref{fig-hidden-capacity}]\label{exist-hidden-channels}
For any fip,
let $r$ be a run, let $i$ be a process and let $m$ be a time such that
$i$ is active at time $m-1$.
Let $c$ be the hidden capacity of $\node{i,m}$ and
let $i_b^{\ell}$, for all $\ell\le m$ and $b=1,\ldots,c$,
be as in \cref{hiddencapacity}.
For every $c$ values $v_1,\ldots,v_c$ of $\Vals$, there exists a run $r'$ of the protocol
such that  $r'_i(m)=r_i(m)$,
and for all $\ell$ and $b$, (a)
$v_b \in \knownvals{i_b^{\ell}, \ell}$
(b) $\knownvals{i_b^{\ell},\ell} \setminus \set{\mathtt{v}_b} \subseteq \knownvals{i,\ell}$,
and (c) $\node{i_b^{\ell},\ell}$
has hidden capacity $\ge c-1$ witnessed by
$i_{b'}^{\ell'}$ for $b'\ne b$ and $\ell'\le \ell$.
\end{lemma}

\begin{proof}
It is enough to define $r'$ up to the end of round $m$.
Let $i_b^{\ell}$,
for all $\ell\le m$ and $b=1,\ldots,c$,
be as in \cref{hiddencapacity}.
We define $r'$ to be the same as $r$, except for the following possible changes
(possible, as they may or may not hold in $r$):
\begin{enumerate}
\item
$i_b^0$ is assigned the initial value $b$, for every $b$.
\item
For every $0\le \ell < m$ and every $b$, the process
$i_b^{\ell}$ fails at $\ell$, at which
it successfully sends a message only to $i_b^{\ell+1}$.
\item
For every $0 < \ell \le m$ and every $b$, the process
$i_b^{\ell}$
receives, until time $\ell-1$ inclusive, the
exact same messages as in $r$. (By definition, $\node{i_b^\ell,\ell-1}$ is seen
by $\node{i,m}$ in $r$, and thus it indeed receives messages in $r$ until time
$\ell-1$, inclusive.) At time $\ell$, the process
$i_b^{\ell}$ receives the exact same messages as $i$, and, in addition,
a message from $i$ and the aforementioned message from~$i_b^{\ell-1}$.
\end{enumerate}

It is straightforward to check, using backward induction on $\ell$, that in $r'$,
each $\node{i_b^{\ell},\ell}$ is not seen up to time $m$ by any process
other than
$i_b^{\ell'}$ for $\ell'>\ell$,
and is thus hidden from $\node{i,m}$ and from $i_{b'}^{\ell'}$ for all
$b' \ne b$ and for all $\ell'$.
Thus, for all $b$ and $\ell$, $\node{i_b^{\ell},\ell}$
has hidden capacity $\ge c-1$ witnessed by
$i_{b'}^{\ell'}$ for $b'\ne b$ and $\ell'\le \ell$.

We now show that none of the above changes alter the state of $i$ at $m$.
By definition, each $\node{i_b^{\ell},\ell}$ is hidden from $\node{i,m}$
in $r$, and as explained above --- in $r'$ as well.
We note that all modifications above affect a process $i_b^{\ell}$ only at
or after time $\ell$, and as this process at these times is not seen by
$\node{i,m}$
in either run, these modifications do not alter the state of $i$
at $m$.

Let $b \in \set{1,\ldots,c}$.
By definition of $r'$, we have $\knownvals{i_b^0,0} = \set{\mathtt{v}_b}$.
Since for every $\ell>0$, $\node{i_b^{\ell},\ell}$ receives a message from
$\node{i_b^{\ell-1},\ell-1}$, we have by induction that
$v_b \in \knownvals{i_b^{\ell},\ell}$ for all $\ell$.

We now complete the proof by showing by induction that for all $\ell$,
$\knownvals{i_b^{\ell},\ell} \subseteq \knownvals{i,\ell} \cup \set{\mathtt{v}_b}$.\footnote{A similar argument to the one used below in fact further shows that
for all $\ell>0$ and for all $b$,
$\knownvals{i_b^{\ell},\ell} =
\knownvals{i,\ell} \cup \set{\mathtt{v}_b}$ in $r'$ for all $\ell$ and $b$.}

Base: $\knownvals{i_b^0,0} = \set{\mathtt{v}_b} \subseteq \knownvals{i,0} \cup \set{\mathtt{v}_b}$.

Step: Let $\ell>0$.
Let $v \in \knownvals{i_b^{\ell},\ell}$.
If $v \in \knownvals{i_b^{\ell},\ell-1}$, then
$v \in \knownvals{i,\ell}$, as $v_b^{\ell}$
is nonfaulty at $\ell-1$ and thus its message is received by $\node{i,\ell}$.
Otherwise, $i_b^{\ell}$ is informed that $\existsv$ by a message it receives
at $\ell$.
By definition of $r'$, a message received by $\node{i_b^{\ell},\ell}$ is exactly
one of the following:
\begin{itemize}
\item
A message received by $\node{i,\ell}$. In this case, $v \in \knownvals{i,\ell}$ as well.
\item
A message sent by $\node{i,\ell-1}$. In this case, we trivially have
$v \in \node{i,\ell-1} \subseteq \knownvals{i,\ell}$.
\item
A message sent by $i_b^{\ell-1}$. In this case, by the induction hypothesis,
\[
v \in \knownvals{i_b^{\ell-1},\ell-1} \subseteq \knownvals{i,\ell-1} \cup \set{\mathtt{v}_b}
\subseteq \knownvals{i,\ell} \cup \set{\mathtt{v}_b}.\]
\end{itemize}
Thus, the proof by induction, and thus the proof of the lemma, is complete.
\end{proof}

\begin{proof}[Proof of \cref{two-face}]
We prove the lemma by induction on $m$.

Base ($m=0$):
Since $K_i \existsv$ at time $0$, the value~$\mathtt{v}$ must be $i$'s initial value, and thus
$\knownvals{i,0}=\set{\mathtt{v}}$.
As $\node{i,m}$ is low, $i$ decides at $0$.
By the \Validity\ property of $P$, it must decide on a value in $\knownvals{i,0}$,
namely, on $\mathtt{v}$.

Step ($m>0$):
Let $i_b^{\ell}$,
for all $\ell\le m$ and $b=1,\ldots,k-1$,
be as in \cref{hiddencapacity}. (See \cref{fig-two-face:first}.)
Let $r'$ be the run of $P$ guaranteed to exist
by \cref{exist-hidden-channels},
with respect to the values $\set{0,\ldots,k-1} \setminus \set{\mathtt{v}}$. (See
\cref{fig-two-face:second}.)
As $j_1,\ldots,j_k$ are seen by $i$ up to time $m$,
we assume w.l.o.g.\ that neither $j_1,\ldots,j_k$ nor $i$ ever fail in $r'$.
We henceforth work in $r'$.

For readability, let us denote by $i_w$, for all
$w \in \set{0,\ldots,k-1} \setminus \set{\mathtt{v}}$, the unique process among
the $i_{b'}^{m-1}$ associated with the value $w$ in the definition of $r'$ by
\cref{exist-hidden-channels}.
Hence,
$w \in \knownvals{i_w,m-1} \cap \set{0,\ldots,k-1} = \knownlows{i_w,m-1}$.
By Condition~\ref{two-face-first-time}, $\node{i,m-1}$ is high,
and thus, by definition of $i_w$,
$\knownlows{i_w,m-1} = \knownvals{i_w,m-1} \cap \set{0,\ldots,k-1}
\subseteq
(\knownvals{i,m-1} \cap \set{0,\ldots,k-1}) \cup \set{w} =
\knownlows{i,m-1} \cup \set{w} = \set{w}.$
We conclude that $\knownlows{i_w,m-1} = \set{w}$.

As $\knownlows{i,m} \setminus \knownlows{i,m-1} = \set{\mathtt{v}} \setminus \emptyset =
\set{\mathtt{v}}$, process~$i$ learned that $\existsv$ by a message it received
at~$m$. Let $i_v$ denote the sender of this message.
We thus trivially have that $\mathtt{v} \in \knownlows{i_v,m-1}$.
Furthermore, we have $\knownlows{i_v,m-1} \subseteq \knownlows{i,m} = \set{\mathtt{v}}$,
and thus $\knownlows{i_v,m-1} = \set{\mathtt{v}}$.

Define $i_k^{m-1}\eqdef i_v$ and $\mathtt{v}_k \eqdef \mathtt{v}$.
As $\knownlows{i_k^{m-1},m-1} = \set{\mathtt{v}}$, for every
$\ell<m-1$ there exists a process $i_k^{\ell}$ s.t.\ (a) $\node{i_k^{\ell},\ell}$ is seen
by $\node{i_k^{\ell+1},\ell+1}$ (and thus does not fail before $\ell$)
and (b) $\mathtt{v} \in \knownlows{i_k^{\ell},\ell}$ (and thus
$\knownlows{i_k^{\ell},\ell} = \set{\mathtt{v}}$). (See \cref{fig-two-face:second}.)
Let $w \in \set{0,\ldots,k-1} \setminus \set{\mathtt{v}}$ and let $\ell<m$.
As $\knownlows{i_w,m-1}=\set{w}$,
and as $\knownlows{i_k^{\ell},\ell}=\set{\mathtt{v}} \ne \set{w}$,
$\node{i_k^{\ell},\ell}$ is not seen by $\node{i_w,m-1}$
and thus (as $i_k^{\ell}$ does not fail before $\ell$), it is hidden from
$\node{i_w,m-1}$.
Furthermore, as $\knownlows{i_k^{\ell},\ell} = \set{\mathtt{v}}$, it is distinct
from all $i_{b}^{\ell}$ for $b<k$.
Let now $w \in \set{0,\ldots,k-1}$.
We conclude that $\node{i_w,m-1}$
has hidden capacity $\ge k-1$ witnessed by
$i_b^{\ell}$ for $\ell \le m-1$ all for all $b$ s.t.\ $\mathtt{v}_b \ne \mathtt{w}$. (See
\cref{fig-two-face:third}.)
Thus, by the induction hypothesis,
$i_w$ decides $\mathtt{w}$ by time $m-1$.

We now apply a sequence of consecutive possible
changes to $r'$ (possible, as they may or may not actually modify $r'$),
numbered from $k$ to $1$. (See \cref{fig-two-face:fourth}.)
For every $b=1,\ldots,k$,
change~$b$ possibly modifies only~$j_b$, and only at times~$\ge m$,
and does not contradict the fact that $i$ and all $j_1,\ldots,j_k$ never
fail.
Therefore, change~$b$ does not affect the state
$i$ or of $j_{b'}$'s up to time $m$, inclusive. Therefore, once change~$b$
is performed, the state of $j_b$ at $m$ is no longer affected by subsequent
changes.
As we show that following change~$b$, $j_b$ decides at $m$,
and denote the value decided upon by $\mathtt{v}_b$, we therefore have that
the fact that $j_b$
decides upon $\mathtt{v}_b$ at $m$ at the latest continues to hold throughout the
rest of the changes.

We now inductively describe the changes (recall that changes are performed starting
with change~$k$ and concluding with change~$1$):
Define $r^k\eqdef r'$. For every $b$, change~$b$ is applied to $r^b$
to yield a run $r^{b-1}$.
Let $b \in \set{1,\ldots,k}$ and assume that changes $k,\ldots,b+1$ were already
performed, and that for each $b'>b$, we have that in $r^{b'-1}$ (and thus
in $r^b$), $j_{b'}$ decides
a low value $\mathtt{v}_{b'}$ by $m$ at the latest, such that $j_{b+1},\ldots,j_k$ are
distinct of each other.

Change~$b$: Let $j_b$ never fail. Furthermore, let $j_b$ receive at time $m$
messages exactly from (a)
$\set{i_0,\ldots,i_{k-1}} \setminus \set{i_{v_{b+1}},\ldots,i_{v_k}}$,
(b)
$i$, and
(c)
$j_1,\ldots,j_k$, except, of course, from $j_b$.

As $i$ and $j_1,\ldots,j_k$ are all high at $m-1$, and as
$\knownlows{i_w,m-1} = \set{w}$ for all $w$, we now have
$\knownlows{j_b,m} = \set{0,\ldots,k-1} \setminus \set{\mathtt{v}_{b+1},\ldots,v_k}$.
In particular, as $b>0$, $\node{j_k,m}$ is low, and therefore must decide
at $m$ or before.
We note that there exists a run $s$ s.t.\ $s_{j_b}(m)=r_{j_b}^{b-1}(m)$,
in which neither $j_b$, nor any of the processes from which it receives messages
at $m$, ever fail. In this run, $j_{b+1},\ldots,j_k$ respectively decide on
$v_{b+1},\ldots,v_k$, and $\set{i_0,\ldots,i_{k-1}} \setminus \set{i_{v_{b+1}},\ldots,i_{v_k}}$ decide on the rest of $\set{0,\ldots,k-1}$. Thus,
by the \kAgreement\ property of $P$,
$j_b$ must decide in $s$ on a value $v_b \in \set{0,\ldots,k-1}$.
As $\knownlows{j_b,m} = \set{0,\ldots,k-1} \setminus \set{\mathtt{v}_{b+1},\ldots,v_k}$,
by the \Validity\ property of $P$, we have that $v_b \ne \set{\mathtt{v}_{b+1},\ldots,v_k}$.
As $s_{j_b}(m)=r_{j_b}^{b-1}(m)$,
$j_b$ must decide on $v_b$ in $r^{b-1}$ as well and the proof by induction is
complete.

By the above construction, $r_i^0(m)=r'_i(m)=r_i(m)$.
Thus, it is enough to show that in $r^0$, $i$ decides on $\mathtt{v}$ at $m$. We thus,
henceforth, work in $r^0$.
As in $r$, and thus also in~$r^0$, $\node{i,m}$ is low, $i$ must decide by
$m$ at the latest.
As all of $j_1,\ldots,j_k$ never fail, and furthermore, collectively
decide on all of $\set{0,\ldots,k-1}$ (see \cref{fig-two-face:fourth}), by the \kAgreement\ property of $P$,
as $i$ never fails, it must decide on a low value.
By the \Validity\ property, $i$ must decide on a value known to it to exist.
As $\knownlows{i,m}=\{v\}$ (in~$r$, and thus also in~$r^0$), we have
that $i$ decides $\mathtt{v}$. As $v \notin \knownlows{i,m-1}$, by \Validity\
we obtain that $i$ does not decide before $m$ and the proof is complete.
\end{proof}

Using \cref{exist-hidden-channels,two-face},
we derive a necessary condition for deciding in $\OptMink$.

\begin{lemma}
\label{k-set-cant-decide-before}
Let $P$ be a protocol solving $k$-set consensus.
Assume that in~$P$, every
undecided low process
must decide.
Then
no high process with hidden capacity $\ge k$ decides in $P$.
\end{lemma}

\begin{proof}
Let $r$ be a run of $P$, let $i$ be a process and let $m$ be a time s.t.\
$\node{i,m}$ is high and has hidden capacity $\ge k$.
Let $i_b^{\ell}$, for all $\ell\le m$ and $b=1,\ldots,k$,
be as in \cref{hiddencapacity}.
Let $r'$ be the run of $P$ guaranteed to exist
by \cref{exist-hidden-channels},
with respect to the values $\set{0,\ldots,k-1}$,
with $i_b^{\ell}$
associated with the value $b-1$ for all $\ell$.
As $i_b^m$, for all $b$, are seen by $i$ up to time $m$,
we assume w.l.o.g.\ that neither they nor $i$ ever fail in $r'$.
As $r'_i(m)=r_i(m)$,
it is enough to show that $i$ does not decide at $m$ in $r'$. We thus,
henceforth, work in $r'$.

Let $b \in \set{0,\ldots,k-1}$. By definition of $r'$,
$\knownlows{i_b^m,m} = \knownvals{i_b^m,m}\cap\set{0,\ldots,k-1} \subseteq
(\knownvals{i,m} \cap \set{0,\ldots,k-1}) \cup \set{b-1} = \knownlows{i,m}
\cup \set{b-1} = \set{b-1}$. As $b-1 \in \knownlows{i_b^m,m}$, we conclude that
$\knownlows{i_b^m,m} = \set{b-1}$.
If $m=0$, then we trivially have that $i_b^m$ is low for the first time at $m$.
Otherwise, as $\node{i,m}$ is high, and as, by definition,
$\node{i_b^m,m-1}$ is seen
by $\node{i,m}$ (in $r$, and therefore in $r'$), we have that $i_b^m$ is
low at $m$ for the first time as well.
By definition of $r'$, $i_b^m$ has hidden capacity $\ge k-1$.
By applying \cref{two-face} with $i$ and $\set{i_{b'}^m}_{b' \ne b}$ as
$j_1,\ldots,j_k$, we thus obtain that $i_b^m$ decides $b-1$ at $m$.

Thus, all of $\set{0,\ldots,k-1}$ are decided upon and so, by the \kAgreement\
property of $P$, $i$ may not decide on any other value.
As $\node{i,m}$ is high, by the \Validity\ property of $P$,
$i$ may not decide on any of $\set{0,\ldots,k-1}$ at $m$.
Thus, $i$ does not decide at $m$.
\end{proof}

\cref{thm:OptMink} follows from \cref{k-set-correct,k-set-cant-decide-before}.

\subsection{A Topological Proof of \texorpdfstring{\cref{two-face}}{Lemma~\ref{two-face}}}
\label{sec-topo-proof}

\subsubsection{Basic Elements of Combinatorial Topology}
\label{sec-topo-def}

A \emph{complex} is a finite set $V$
and a collection of subsets $\cK$ of $V$ closed under containment.
An element of $V$ is called a \emph{vertex} of $\cK$,
and a set in $\cK$ is called a \emph{simplex}.
A (proper) subset of a simplex $\sigma$ is called a \emph{(proper) face}.
The \emph{dimension} $\dim \sigma$ is $|\sigma|-1$.
The dimension of a complex $\cK$, $\dim \cK$,
is the maximal dimension of any of $\cK$'s simplexes.
A complex $\cK$ is \emph{pure} if all its simplexes have the same dimension.

Let $\cK$ be a complex and $v$ one of its vertices.
The \emph{star complex} of $v$ in $\cK$, denoted $\Star(v, \cK)$
is the subcomplex of $\cK$ containing every simplex, and all its faces,
that contains $v$.

For a simplex $\sigma$, let $\bdry \sigma$ denote the
complex containing all proper faces of $\sigma$.
If $\cK$ and $\cL$ are disjoint, their \emph{join}, $\cK \ast \cL$,
is the complex  $\set{ \sigma \cup \tau : \sigma \in \cK \wedge \tau \in \cL}$.

A \emph{coloring} of a complex $\cK$ is a map from the vertices of
$\cK$ to a set of \emph{colors}.
A simplex of $\cK$ is \emph{fully colored} if its vertices are mapped to distinct colors.

Informally, a \emph{subdivision} $\Div \sigma$ of $\sigma$ is a complex
constructed by subdividing each $\sigma' \subseteq \sigma$ into smaller simplexes.
A subdivision $\Div \sigma$ maps each $\sigma' \subseteq \sigma$ to the
pure complex $\Div \sigma'$ of dimension $\dim \sigma$ containing the simplexes that subdivide $\sigma'$.
Thus, for all $\sigma', \sigma'' \subseteq \sigma$,
$\Div \sigma' \cap \Div \sigma'' = \Div \sigma' \cap \sigma''$.
For every vertex $v \in \Div \sigma$,
its \emph{carrier}, $\Car v$,
is the face $\sigma' \subseteq \sigma$ of smallest dimension such that
$v \in \Div \sigma'$.

The \emph{barycentric} subdivision $\Bary \sigma$ of $\sigma$ can be
defined in many equivalent ways. Here we adopt the following
combinatorial definition.
$\Bary \sigma$ is defined inductively by dimension.
For dimension 0, for every vertex $v$ of $\sigma$,
$\Bary v = v$.
For dimension $\ell$, $1 \leq \ell \leq \dim \sigma$,
for every $\ell$-face $\sigma'$ of $\sigma$,
for a new vertex $v = \sigma'$,
$\Bary \sigma' = v \ast \Bary \bdry \sigma'$.

Let $\Div \sigma$ be a subdivision of $\sigma$.
A \emph{Sperner coloring} of $\Div \sigma$ is a coloring
that maps every vertex $v \in \Div \sigma$ to a vertex in $\Car v$.

\subsubsection{Proof of \texorpdfstring{\cref{two-face}}{Lemma~\ref{two-face}}}

Consider any run $r$. We proceed by induction on the time $m$.

For the base of the induction $m=0$, if the four conditions holds for
a process $i$ at time $0$, then it must be that
$i$ starts in $r$ with input $\mathtt{v}$,
and consequently $V\ang{i, m} = \set{\mathtt{v}}$.
Therefore, $i$ decides $\mathtt{v}$ at time $0$, since $P$ satisfies the validity requirement
of $k$-set consensus.

Let us assume the claim holds until time $m-1$. We prove it holds at $m$.
Let $i$ be a process that satisfies the four conditions at time $m$.

Without loss of generality, let us assume $L\ang{i,m}=\{0\}$.
Let $i_0$ be a process such that $i$ receives a message
from $i_0$ at time $m$ and $L\ang{i_0,m-1} = \set{0}$.
We have $\ang{i,m}$ has hidden capacity greater or equal than $k-1$,
thus,
\cref{exist-hidden-channels}
implies that
there exist a run $r'$ indistinguishable to $\ang{i,m}$ such that
there are $k-1$ processes $i_1, \hdots, i_{k-1}$
such that for each $i_x$, $1 \leq x \leq k-1$,
$\ang{i_x, m-1}$ hidden to $\ang{i,m}$
and $L\ang{i_x,m-1} = \set{x}$.

By induction hypothesis, every $i_x$, $0 \leq x \leq k-1$,
decides at time $m-1$, at the latest, on its unique low value~$x$.
We assume, for the sake of contradiction, that $i$ decides
on a non-low value at time~$m$
(if $i$ decides before, it necessarily decides on a non-low value).
For simplicity, let us assume $i$ decides on $k$.

By hypothesis, there are $k$ processes, $j_1, \hdots, j_k$,
(distinct from $i$ and $i_x$)
such for each $1 \leq y \leq k$, $L\ang{j_y, m-1} = \emptyset$.
Note that $k \in H\ang{j_y, m-1}$, for every $j_y$.

Below, we only consider runs in which
a subset of $i_0, \hdots, i_{k-1}$ crash in round $m$
and every $j_y$ receives at least one message from some $i_x$;
all other process do not crash in round $m$.
Thus, $L\ang{j_y, m} \neq \emptyset$, for every $j_y$,
and consequently it decides at time $m$, at the latest.

We now define a subdivision, $\Div \sigma$, of a $k$-simplex
$\sigma = \{ 0, \hdots, k \}$,
and then define a map $\delta$ from the vertices $\Div \sigma$ to
states of  $i$, $i_x$ or $j_y$ at time~$m$.
The mapping $\delta$ will be defined in a way that
the decisions of the processes induce a Sperner coloring on $\Div \sigma$.
Finally, we argue that,
for every simplex $\tau \in \Div \sigma$,
all its vertices are mapped to distinct compatible process states
in some execution.
Therefore, by Sperner's Lemma, there must be a $k$-dimensional simplex in $\Div \sigma$
in which $k+1$ distinct values are decided by distinct processes,
thus reaching a contradiction.

\begin{lemma}[Sperner's Lemma]
Let $\Div \sigma$ be a subdivision with a Sperner coloring $\zeta$.
Then, $\zeta$ defines an odd number of fully colored $(\dim \sigma)$-simplexes.
\end{lemma}

We construct $\Div \sigma$ inductively by dimension.
The construction is a simple variant of the well-known barycentric subdivision
(see \cref{fig-subdivision} (left)).

For dimension $0$, for every vertex $v \in \sigma$, we define $\Div v = v$;
hence $\Car v = v$.
For every $1$-face (edge) $\sigma'$ of $\sigma$,
if $k \notin \sigma'$ or $\sigma' = \{0, k\}$, then $\Div \sigma' = \sigma'$;
otherwise, for a new vertex $v = \sigma'$,
$\Div \sigma' = v \ast \Div \bdry \sigma'$.
Note that $\Car v = \sigma'$.
For every $x$-face $\sigma'$ of $\sigma$, $2 \leq x \leq k$,
if $k \notin \sigma'$, then $\Div \sigma' = \sigma'$;
otherwise, for a new vertex $v = \sigma'$,
$\Div \sigma' = v \ast \Div \bdry \sigma'$.
Again note that $\Car v = \sigma'$
(see \cref{fig-subdivision} (center)).

\begin{figure}[t]
\centering
        \includegraphics[scale=.8]{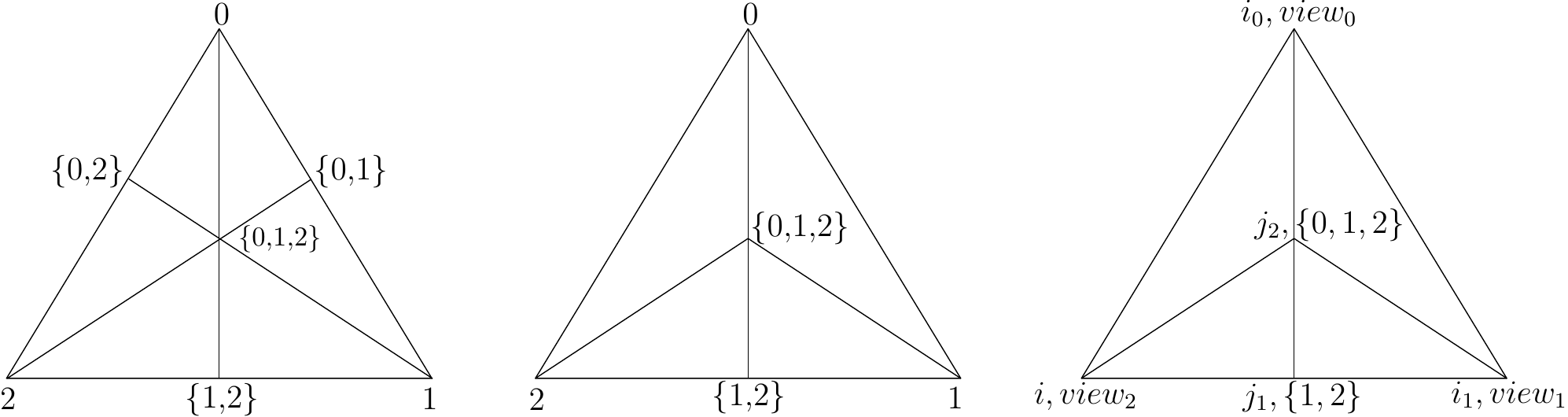}
        \caption{\footnotesize
        For dimension $k=2$ and $\sigma = \set{0,1,2}$,
        the barycentric subdivision $\sigma = \set{0,1,2}$ appears at the left,
		while the subdivision $\Div \sigma$ appears at the center.
        In the subdivision at the right, the vertices are mapped to process states.
        For example, the triangle $\set{\ang{i_0,view_0}, \ang{i_1, view_1}, \ang{j_2, \set{0,1,2}}}$ corresponds
        to the execution in which $i_0$ and $i_1$ do not crash in round $m$, and hence
        $j_1$ receives $0$ and $1$, which are included in its view.
   		Similarly, the triangle $\set{\ang{i_1, view_1}, \ang{j_1, \set{1,2}}, \ang{j_2, \set{0,1,2}}}$
   		correspond to the execution in which $i_1$ does not crash in round $m$,
   		while $i_0$ crashes and sends a message to $j_2$ and no message to $j_1$.
        The decisions of the processes induces a Sperner coloring:
        by assumption, $i$ decides $k=2$, and by induction hypothesis,
        $i_0$ and $i_1$ decided $0$ and $1$ at time $m-1$;
        the rest of the processes have to decide at time $m$ and
        they can only decide values in their views.
        }
		\label{fig-subdivision}
\end{figure}

We now define the mapping $\delta$ and the Sperner coloring of $\Div \sigma$,
which is induced by the decision function $\zeta$ of $P$.

For every vertex $v \in \sigma$, $\Div v = v$.
If $v \neq k$, then $\delta(v)$ is the state $\ang{i_v,m}$ in which
$i_v$ sends and receives all its messages,
i.e. $i_v$ does not crash in round $m$;
otherwise, $\delta(v) = \ang{i,m}$ in $r'$.
Note that for all $v \in \sigma$, $\zeta(\delta(v)) = v$,
by induction hypothesis and because we assume $i$ decides on $k$.

For every, $y$-face $\sigma'$ of $\sigma$, $1 \leq y \leq k$,
if there is a vertex $v \in \Div \sigma'$ with $\Car v = \sigma'$,
then $v = \sigma'$ and $k \in \sigma'$.
For such a vertex, we define $\delta(v)$ to be the state $\ang{j_y,m}$ in which
(a) $j_y$ receives a message from $i_w$, for every $w \in \sigma'$
($i_w$ may crash after sending a message to $i_y$), and
(b) $j_y$ does not receive any message from the $i_{x}$'s whose
subindexes do not appear in $\sigma'$,
namely, they crash in round $m$
without sending a message to $j_y$.
Observe that $L\ang{i_y,m} = \sigma' \setminus \set{k}$
and $H\ang{i_y,m}$ contains $k$ and possible more high values distinct from $k$.
Since $P$ satisfies the validity requirement
of $k$-set consensus,
$\zeta(\delta(v))$ is any value in $V\ang{i_y,m} = L\ang{i_y,m} \cup H\ang{i_y,m}$.
For now, we assume that if $\zeta(\delta(v)) \in H\ang{i_y,m}$,
then $\zeta(\delta(v)) = k$, in other words,
if $i_y$ decides a high value, it decides on $k$;
hence $\zeta(\delta(v)) \in \Car v$.
Therefore, $\zeta$ defines a Sperner coloring for $\Div \sigma$.
Later we explain that this assumption does not affect our argument below.

Consider a $k$-simplex $\tau \in \Div \sigma$.
To show that $\delta$ maps the vertices of $\tau$
distinct process states, it is enough to see that
for every $v \in \Div \sigma$,
if $\dim \Car v = 0$, then $\delta(v)$ is a state of $i$ or some $i_x$;
and if $1 \leq \dim \Car v \leq k$,
then $\delta(v)$ is a state of $j_y$, where $y = \dim \Car v$.
And to show that $\delta$ map $\tau$ to states of an execution,
note that if there is a $v \in \tau$ such that
$\delta(v) = \ang{i,m}$, then
the states in $\delta(\tau)$
correspond to an execution in which
each $j_y$ receives a subset of the messages from
$i_0, \hdots, i_{k-1}$;
otherwise, the states in $\delta(\tau)$
correspond to an execution in which some
$i_x$'s distinct from $i_0$
do not crash in round $m$ (see \cref{fig-subdivision} (right)).
Observe that in the second case,
the state of $i$ at time $m$ in that execution is different from
the state of $i$ at time $m$ in $r'$, because in $r'$ $i$ only receives
a message from $i_0$.

By Sperner's Lemma, there is at least one fully colored $k$-simplex in $\Div \sigma$,
and thus there is an execution of $P$ in which $k+1$ distinct values are decided at time $m$.
A contradiction.

Finally, we assumed that if
$i_y$ decides a high value, it decides on $k$.
Observe that if in $\Div \sigma$, we replace
$k$ with the actual decision of $i_y$, then,
the number of distinct decision  at the vertices
of a simplex of $\Div \sigma$ can only increase.
Thus, in any case, $\Div \sigma$ has a simplex with
$k+1$ distinct decisions. The lemma follows. \qed

\subsection{
Unbeatability and Connectivity}
\label{sec-discussion-connectivity}

\begin{proof}[Proof of \cref{hidden-capacity-connectivity} (Sketch)]
The lemma can be proved using similar techniques as in~\cite{GHP,HRT98}.
Roughly speaking, given a full-information protocol, the analysis in those paper considers the $m$-round protocol complex ${\cal P}_m$
containing the executions in which there are at most $k$ failures per round.
Then, after an elaborate analysis, it is shown that ${\cal P}_m$ is $(k-1)$-connected.
The main idea is that the number of failures in the execution modeled in ${\cal P}_m$
have a ``level'' of uncertainty which is captured by the $(k-1)$-connectivity of ${\cal P}_m$.
Here it is worth to note that if a process sees $k$ new failures in every round, its
hidden capacity is $k$, however, that is not the only scenario in which hidden capacity
can be $k$, it depends how information flows during a given execution.

Observe that the connectivity property of ${\cal P}_m$ is ``global''.
In contrast, in this lemma we focus on a local property since we consider the star complex
of a vertex $v$, $\Star(v, {\cal P}_m)$. However, the principle is the same: the ``level'' of uncertainty of $v$,
namely, in every round it has hidden capacity at least $k$, is reflected in the connectivity
of its star complex.

The claim can be proved by induction on the number of rounds (with round operators),
using similar techniques that have been used in the past.
\end{proof}

\section{Uniform Set Consensus}
\label{sec-uni-set-cons-proofs}

In the following proof, following \cite{AYY-DISC}, we denote the fact that process~$i$ (at time~$m$) knows that $\valv$ will persist by $K_i\cv$; the time $m$ will be clear from context.

\begin{proof}[Proof of \cref{u-k-solve}]

\Decision:
By definition of $\UOptMink$, every process that is active at time $\lfloor\nicefrac{\tee}{k}\rfloor+1$,
and in particular every nonfaulty process, decides by this time at the latest.

Before moving on to show \Validity\ and \UnikAg, we first complete the analysis of stopping times.
In some run of $\UOptMink$, let $i$ be a process and let $m$ be a time s.t.\ $i$ is active at $m$ but has not decided until $m$, inclusive.
Let $\tilde{m}\le m$ be the latest time not later than $m$ s.t.\ $\node{i,\tilde{m}}$ has hidden capacity $\ge k$. By definition
of $\UOptMink$, as $i$ is undecided at $m$, we have $\tilde{m} \ge m-1$.

As $\node{i,\tilde{m}}$ has hidden capacity $\ge k$ at $\tilde{m}$, let $i_b^{\ell}$, for all $0\le\ell\le \tilde{m}$ and $b=1,\ldots,k$, be as in \cref{hiddencapacity}.
By definition, $\node{i_b^{\ell},\ell}$, for every $0\le\ell < \tilde{m}$ and $b=1,\ldots,k$, is hidden from $\node{i,\tilde{m}}$.
Thus, $k\cdot \tilde{m} \le \knownf{i,\tilde{m}} \le f$.
therefore, $\tilde{m} \le \nicefrac{f}{k}$ and so $\tilde{m} \le \lfloor\nicefrac{f}{k}\rfloor$.
Hence, as $m-1 \le \tilde{m}$, we have $m \le \tilde{m}+1 \le \lfloor \nicefrac{f}{k} \rfloor + 1$.
We thus have that every process that is active at time $\lfloor\nicefrac{f}{k}\rfloor+2$, decides by this time at the latest.

We move on to show \Validity\ and \UnikAg. Henceforth,
let $i$ be a (possibly faulty) process that decides in some run of $\UOptMink$, let $m_i$ be the decision time of $i$, and
let $\mathtt{v}$ be the value upon which $i$ decides.
Thus, there exists $m'_i \in \{m_i,m_i-1\}$ s.t.\ $\node{i,m'_i}$ is low or has hidden capacity $<k$, and s.t.\ $\mathtt{v} = \minval{i,m'_i}$.
(To show this when $m_i=\lfloor\nicefrac{\tee}{k}\rfloor+1$, we note that in this case $m_i>\lfloor\nicefrac{f}{k}\rfloor$, and so, as shown in the stopping-time analysis above, this implies that $\node{i,m_i}$ has hidden capacity $<k$.)

\Validity:
As $\mathtt{v}=\minval{i,m'_i}$, we have $K_i\existsv$ at $m'_i$, and thus $\existsv$.

\UnikAg:
It is enough to show that at most $k-1$ distinct values
smaller than $\mathtt{v}$ are decided upon in the current run.
If $\node{i,m'_i}$ is low, then $\mathtt{v} = \minval{i,m'_i} < k-1$, and thus there do not exist more than $k-1$ distinct legal values smaller than $\mathtt{v}$, let alone ones decided
upon. For the rest of this proof we assume, therefore, that $\node{i,m'_i}$ is high, and so has hidden capacity $<k$.

Let $\mathtt{w}<\mathtt{v}$ be a value decided upon by some process. Let $j$ be this process, and let~$m_j$ be the time at which $j$ decides on $\mathtt{w}$.
Thus, $\mathtt{w} = \minval{j,m'_j}$ for some $m_j' \in \{m_j,m_j-1\}$ s.t.\ if $m'_j=m_j$, then either $K_j\cw$ at $m_j$, or $m_j=\lfloor\nicefrac{\tee}{k}\rfloor+1$.

We first show that $m'_j \ge m'_i$. If $m'_j=m_j$ and $m_j=\lfloor\nicefrac{\tee}{k}\rfloor+1$, then we immediately have
$m'_j=\lfloor\nicefrac{\tee}{k}\rfloor+1\ge m_i \ge m'_i$, as required. Otherwise, the analysis is somewhat more subtle. We
first show that in this case, if $i$ is active at $m'_j+1$, then $K_i\exists\mathtt{w}$ at $m'_j+1$. We reason by cases, according to the value of $m'_j$.
\begin{itemize}
\item
If $m'_j=m_j$, then $K_j\cw$ at $m'_j$, and thus there exists a process $k$ that never fails,
s.t.\ $K_k\exists\mathtt{w}$ at $m'_j$. As $k$ never fails, $\node{k,m'_j}$ is seen by $\node{i,m'_j+1}$, and thus $K_i\exists\mathtt{w}$ at $m'_j+1$, as required.
\item
Otherwise, $m'_j=m_j-1$. As $j$ is active at $m_j$, it does does not fail at $m'_j<m_j$, and therefore $\node{j,m'_j}$ is seen by
$\node{i,m'_j+1}$. Thus, as $K_j\exists\mathtt{w}$ at $m'_j$, we obtain that $K_i\exists\mathtt{w}$ at $m'_j+1$ in this case as well.
\end{itemize}
As $\mathtt{w}<\mathtt{v}$ and as $\mathtt{v} = \minval{i,m'_i}$, we have $\lnot K_i\exists\mathtt{w}$ at $m'_i$. Thus, we obtain that $m'_i < m'_j+1$, and therefore
$m'_j \ge m'_i$ in this case as well, as required. We have thus shown that we always have $m'_j \ge m'_i$.

As $\node{i,m'_i}$ does not have hidden capacity $k$, there exists $0\le\ell\le m'_i$ s.t.\ no more than $k-1$ processes at time $\ell$ are hidden from~$\node{i,m'_i}$.
As $m'_i \ge \ell$, we have $m'_j \ge m'_i \ge \ell$.
Let $H$ be the set of all processes seen at $\ell$ by $\node{j,m'_j}$. (Note
that if $m'_j = \ell$, then $H = \set{j}$.)
Since $m'_j \ge \ell$,
we have $\knownvals{j,m'_j} = \bigcup_{h \in H} \knownvals{h,\ell}$. Thus, $\mathtt{w}=\minval{j,m'_j}=\min_{h \in H}\{\minval{h,\ell}\}$. Therefore,
$\mathtt{w} = \minval{h,\ell}$ for some $h \in H$. As $\lnot K_i \exists\mathtt{w}$ at $m'_i$, we thus have that $\node{h,\ell}$ is not seen
by $\node{i,m'_i}$. As $\node{h,\ell}$ is seen by $\node{j,m'_j}$, $h$ does not fail
before $\ell$, and thus $\node{h,\ell}$ is hidden from $\node{i,m'_i}$.
To conclude, we have shown that
\[w \in \bigl\{ \minval{h,\ell} \mid
\mbox{$\node{h,\ell}$ is hidden from $\node{i,m'_i}$} \bigr\}.\]
As there are at most $k-1$ processes hidden at $\ell$ from $\node{i,m'_i}$,
we conclude that no more than $k-1$ distinct values lower than $\mathtt{v}$ are
decided upon, and the proof is complete.
\end{proof}

\section{Last-Decider Unbeatability}
\label{sec-last-decider}

We first formally define last-decider unbeatability.

\begin{definition}[Last-Decider Domination and Unbeatability]
\leavevmode
\begin{itemize}
\item
A decision protocol $Q$ \defemph{last-decider dominates} a protocol~$P$ in~$\gamma$, denoted by $Q\boldsymbol{\overset{\smash{l.d.}}{\dom}_\gamma} P$ if, for all adversaries $\alpha$, if $i$ the last decision in~$P[\alpha]$ is at time $m_i$, then all decisions in $Q[\alpha]$ are taken before or at $m_i$. Moreover, we say that $Q$  \defemph{strictly last-decider dominates} $P$
if $Q\overset{\smash{l.d.}}{\dom}_\gamma P$ and  $P\!\!\boldsymbol{\not}\!\!\!\overset{\smash{l.d.}}{\dom}_\gamma Q$. I.e., if for some $\alpha\in\gamma$ the last decision in $Q[\alpha]$ is {\em strictly before} the last decision in $P[\alpha]$.
\item
A protocol $P$ is a \defemph{last-decider unbeatable} solution to a decision task~$S$ in a context~$\gamma$ if $P$ solves~$S$ in~$\gamma$ and no protocol $Q$ solving~$S$ in~$\gamma$ strictly last-decider dominates~$P$.
\end{itemize}
\end{definition}

\begin{remark}
\leavevmode
\begin{itemize}
\item
If $Q\boldsymbol{\dom_\gamma} P$, then $Q\boldsymbol{\overset{\smash{l.d.}}{\dom}_\gamma} P$. (But not the other way around.)
\item
None of the above forms of strict domination implies the other.
\item
None of the above forms of unbeatability implies the other.
\end{itemize}
\end{remark}

Last-decider domination does not imply domination in the sense of the rest of this paper (on which our proof is based).
Nonetheless, the specific property of protocols dominating $\OptMink$ that we use to prove that these protocols are unbeatable, holds also for protocols that only last-decider dominate these protocols.

\begin{lemma}\label{last-dom-sufficient}
Let $Q\overset{\smash{l.d.}}{\dom}\OptMink$ satisfy \Decision. If $i$ is low at $m$ in a run $r\!=\!Q[\alpha]$ of $Q$, then $i$ decides in $r$ no later than at $m$.
\end{lemma}

The main idea in the proof of \cref{last-dom-sufficient} is to show that $i$ considers it possible that
all other active processes also know the fact stated in that part,
and so they must all decide by the current time in the corresponding run of the dominated protocol. Hence, the last decision decision in that run is made in the current time; thus, by last-decider domination, $i$ must decide.

\begin{proof}[Proof of \cref{last-dom-sufficient}]
If $m\!=\!0$, then there exists a run $r'\!=\!Q[\beta]$ of $Q$, s.t.~~\emph{i)} $r'_i(0)\!=\!r_i(0)$,~~\emph{ii)} in $r'$ all initial values are $0$, and~~\emph{iii)} $i$ never fails in $r'$. Hence, in $\OptMink[\beta]$ all decisions
are taken at time $m\!=\!0$, and therefore so is the last decision. Therefore, the last decision in $r'$ must be taken at time $0$. As $i$ never fails in $r'$, by \Decision\ it must decide at some
point during this run, and therefore must decide at $0$ in $r'$. As $r_i(0)\!=\!r'_i(0)$, $i$ decides at $0$ in $r$ as well, as required.

If $m\!>\!0$, then there exists a process $j$ that is low at $m-1$ in $r$ and $\node{j,m-1}$ is seen by $\node{i,m}$. Thus, there exists a run $r'\!=\!Q[\beta]$ of $Q$,
s.t.~~\emph{i)} $r'_i(m)\!=\!r_i(m)$, and~~\emph{ii)} $i$ and $j$ never fail in $r'$. Thus, all processes that are active at $m$ in $r'$ see $\node{j,m-1}$ and are therefore low.
Hence, in $\OptMink[\beta]$ all decisions are taken by time $m$, and therefore so is the last decision. Therefore, the last decision in $r'$ must be taken no later than at time $m$.
As $i$ never fails in $r'$, by \Decision\ it must decide at some point during this run, and therefore must decide by $m$ in $r'$. As $r_i(m)\!=\!r'_i(m)$, $i$ decides by $m$ in $r$ as well, as required.
\end{proof}

As explained above, \cref{thm:last-decider} follows from \cref{last-dom-sufficient}, and from the proof of \cref{thm:OptMink}.

\section{Efficient Implementation}\label{sec-com-eff}

Throughout this paper we have assumed that processes follow the \fip, and did not concern ourselves with implementation details. Notice that the only information that processes use in these protocols concerns the values that processes have seen and the failures that they observe. This determines which nodes are seen by, known to be crashed by, or hidden from a node $\node{i,m}$. We now show that a more efficient protocol exists in which the processes obtain the same information about these three aspects of the run as in the \fip.

\begin{lemma}
For each of the protocols $\OptMink$ and $\UOptMink$, there is a protocol with identical decision times for all adversaries, in which every process sends at most $O(n\log n)$ bits overall to each other process.
\end{lemma}

\begin{proof}[Proof (Sketch)]
Moses and Tuttle in \cite{MT} show how to implement full-information protocols in the crash failure model with linear-size messages. In our case, a further improvement is possible, since decisions in all of the protocols depend only on the identity of hidden nodes and on the vector of initial values. In a straightforward implementation, we can have a process $i$ report  ``{\tt value}$(j) = v$'' once for every $j$ whose initial value it discovers, and ``{\tt failed\_at}$(j) = \ell$'' once where~$\ell$ is the earliest failure round it knows for $j$. In addition, it should send an ``{\tt I'm\_alive}'' message in every round in which it has nothing to report. Process~$i$ can send at most one {\tt value} message and two
{\tt failed\_at} messages for every $j$. Since {\tt I'm\_alive} is a constant-size message sent fewer than~$n$ times, and since encoding~$j$'s ID requires $\log n$ bits, a process~$i$ sends a total of $O(n \log n)$ bits overall.
\end{proof}

We also note that it is straightforward to compute the hidden capacity of a node $\node{i,m}$
in a run with adversary~$\alpha$ based on the communication graph~$\CG_\alpha$. The hidden capacity
of $\node{i,m}$ can also be very efficiently calculated from the hidden capacity of $\node{i,m\!-\!1}$ using auxiliary
data calculated during the calculation of the latter.

\end{document}